\documentclass{jssc}

\def\textsubscript#1%
{$_{\text{#1}}$}

\def\cdd{\mbox{\boldmath$\cdot$}~}

\input{amssym.def}
\include{graphix}

\usepackage{amsmath,amscd}

\usepackage[ruled,vlined,linesnumbered]{algorithm2e}
\usepackage{tikz}

\usepackage{lhelp}
\usepackage{bbm}
\usepackage{hyperref}
\usepackage{amsfonts}
\usepackage{qtree}

\def\C {\ensuremath{\mathbb{C}}}

\def\Q {\ensuremath{\mathbb{Q}}}
\def\R {\ensuremath{\mathbb{R}}}

\renewcommand{\dim}[1]{\mbox{{\rm dim}$(#1)$}}

\newcommand{\Jac}{\mathcal{J}}

\makeatletter
\def\@oddfoot{\hfill}
\newcount\shumeicount
\def\setshumei#1#2#3{%
  \shumeicount=\count0
  \def\@oddhead{%
    \raise-5pt\hbox to0pt{\vrule width\hsize height 0pt depth 0.4pt\hss}\relax
    \ifnum \shumeicount=\count0
      \raise-7pt\hbox to0pt{\vrule width\hsize height 0pt depth 0.4pt\hss}\relax
      #1
    \else
      \ifodd\count0
        #2
      \else
        #3
       \fi
     \fi
  }%
}
\makeatother
\makeatletter
\def\@oddfoot{\hfill}
\newcount\shujiaocount
\def\setshujiao{%
  \shujiaocount=\count0
  \def\@oddfoot{%
      \ifodd\count0
      \else
      \fi
  }%
}
\makeatother
\def\title#1#2#3#4{{
  \vspace*{0.3cm}
  \begin{flushleft} \Large\bf #1\end{flushleft}
  \vspace*{-0.2cm}
      \begin{flushleft}
      \bf #2
      \end{flushleft}
      \footnotetext{\hspace{-6mm} #3\\ #4}}}

\def\dshm#1#2#3#4
{\setshumei{(#1) #2:
{\thepage--\pageref{LastPage}}\hfill}
            {\hfill {\small #3}\hfill\hbox to0pt{\hss\thepage}}
            {\hbox to0pt{\thepage\hss}\hfill {\small #4}\hfill
            }
            \setshujiao}
\def\drd#1#2
{{\vskip 1cm\small \begin{flushleft}
 #1 \\
 #2\\
\end{flushleft}}}

\def\tilde{\widetilde}

\def\epsilon{\varepsilon}

\begin{document}

\title{Visualizing Planar and Space Implicit Real Algebraic Curves with Singularities$^*$}
{\uppercase{Chen} Changbo$^{a,b}$ \cdd \uppercase{Wu}
Wenyuan$^{a,b}$ \cdd \uppercase{Feng} Yong$^{a,b}$}
{
  Address: $^a$Chongqing Key Laboratory of Automated Reasoning and Cognition, Chongqing Institute of Green
  and Intelligent Technology, Chinese Academy of Sciences, Chongqing 400714, China\\
  $^b$University of Chinese Academy of Sciences, Beijing 100049, China\\
    Email: chenchangbo@cigit.ac.cn, wuwenyuan@cigit.ac.cn, yongfeng@cigit.ac.cn} 
{$^*$This research was supported by  NSFC (11771421, 11671377, 61572024),
  CAS ``Light of West China'' Program, 
cstc2018jcyj-yszxX0002 of Chongqing
and the Key Research Program of Frontier Sciences of CAS (QYZDB-SSW-SYS026).\\
}



\dshm{20XX}{XX}{A TEMPLATE FOR JOURNAL}{\uppercase{Chen Changbo} $\cdd$ \uppercase{Wu Wenyuan} $\cdd$
\uppercase{Feng Yong}}

\Abstract{We present a new method for visualizing implicit real algebraic curves  inside a bounding box
  in the $2$-D or $3$-D ambient space based
  on numerical continuation and critical point methods.
  The underlying techniques work also for tracing space curve in higher-dimensional space.
  Since the topology of a curve near a singular point of it is not numerically stable,
  we trace only the curve outside neighborhoods of singular points
  and replace each neighborhood simply by a point,
  which produces a polygonal approximation that is $\epsilon$-close
  to the curve.
  Such an approximation is more stable for defining
  the numerical connectedness of the complement of the projection of the curve in $\R^2$,
  which is important for applications such as solving bi-parametric polynomial systems.
  The algorithm starts by computing three types of key points of the curve,
  namely the intersection of the curve with small spheres centered at singular points,
  regular critical points of every connected components of the curve,
  as well as intersection points of the curve with the given bounding box.
  It then traces the curve starting with and in the order of the above three types of points.
  This basic scheme is further enhanced by several optimizations, such as grouping singular points in natural clusters,
  tracing the curve by a try-and-resume strategy and handling ``pseudo singular points''.
  The effectiveness of the algorithm is illustrated by numerous examples.
  This manuscript extends our preliminary results that appeared in CASC 2018.
}      


\section{Introduction}
Visualizing an implicit plane or space real algebraic curve is a classical and fundamental
problem in computational geometry and computer graphics.
There have been many works on this topic both in $2$-D~\cite{Bresenham1977,Chandler1988,Hong1996,Lopes2002,Sagraloff2009,Labs2010,Burr2012,Jin2015}
and $3$-D~\cite{Daouda2008,Jin2014,DBLP:phd/dnb/Stussak13} cases.
In the literature, a correct visualization usually requires two conditions:
$(i)$ the generated polygonal approximation is $\epsilon$-close to the curve,
and $(ii)$ the approximation is ``topologically correct'', which often means
that the approximation is isotopic to the curve.
There are also many works~\cite{Daouda2008,Cheng2010,Seidel2005,Imbach2017} focusing only on $(ii)$.

Different techniques~\cite{Gomes2014,DBLP:journals/toms/BrakeBHHSW17} exist for visualizing plane and space curves,
such as implicit-to-parametric conversion, curve continuation and space subdivision.
Symbolic or hybrid symbolic-numeric approaches stand out for being capable of
computing the exact topology and many of them are variants of cylindrical algebraic decomposition.
For continuation based approach, several difficulties arise,
such as finding at least one seed point from each connected component, dealing with curve jumping
and handling singularities.
Each of the three problems has its own interests.
For instance, there are several approaches for computing at least one witness points for a real variety,
either symbolically~\cite{col75,Rouillier2000,Chen2013} or numerically~\cite{Hauenstein2012,WRF2017,Wu2017}.
Different techniques for robustly tracing curves are proposed~\cite{Blum1997,Beltran2013,Martin2013,Yu2014,WRF2017}. 
Techniques for handling singularities also exist~\cite{Bajaj1997,Leykin2006,Chen2017}.

For curves with singularities, observe that
condition $(ii)$ is numerically ill-posed, since a slight perturbation may completely
change the topology of the curve nearby a singular point, see Example~\ref{ex:per} for instance.
On the other hand, in many applications, such as solving bi-parametric polynomial systems,
condition $(ii)$ is unnecessary.
Let us illuminate this point now.
For a given bi-parametric polynomial system, one can compute a border curve~\cite{Chen2016,Chen2016b},
or a border polynomial~\cite{YX05} or discriminant variety~\cite{LazardRouillier2007} in general
in the parametric space, where the complement of the curve is a disjoint union of connected open cells,
such that above each cell the number of solutions of the system is constant and
the solutions are continuous functions of parameters with disjoint graphs.
Let $B$ be a border curve and $\tilde{B}$ be a polygonal approximation $\epsilon$-close to it.
In~\cite{Chen2016}, we introduced the notion of $\epsilon$-connectedness and showed that
two points are $\epsilon$-connected w.r.t. $\tilde{B}$ implies that they are connected w.r.t. $B$,
which in turn implies that the parametric system has the same number of solutions at the two points.
Thus an $\epsilon$-approximation of the border curve meeting only condition $(i)$ is good enough
for the purpose of solving parametric systems.
The curve tracing subroutine in~\cite{Chen2016} relies on perturbation to handle singularities.
In this work, we develop a perturbation free algorithm.
The algorithm traces only the curve outside neighborhoods of singular points and replace each neighborhood simply by a point.
Such produced approximation is more stable for defining
the numerical connectedness of the complement of the curve
than those approximations preserving the topology around singular points.

The paper is organized as follows.
In Section~\ref{sec:theory}, we formalize the problem and provide a theoretical base algorithm
to guarantee $\epsilon$-closeness based on a robust curve tracing method.
Several strategies for improving the numerical stability of tracing is proposed in Section~\ref{sec:optimization},
such as tracing the curve away from the singular points rather than towards it,
tracing the curve by a  try-and-resume strategy,
classifying singular points into natural clusters~\cite{Bennett2016},
and handling ``pseudo singular points''.
The theoretical algorithm may require the step size to be very small.
In Section~\ref{sec:algo},  we present a more practical algorithm based on
optimizations in Section~\ref{sec:optimization}.
Instead of preventing curve jumping, it maintains a simple data structure
to detect curve jumping.
The effectiveness of the algorithm is  demonstrated through several nontrivial examples in Section~\ref{sec:exp}.
Finally, in Section~\ref{sec:con}, we draw the conclusion and point out some possible future directions
to improve the current work.
Compared with our CASC 2018 paper~\cite{ChenW18}, our algorithm is now fully generalized to plot curve
in $n$-dimensional ambient space. In particular, its effectiveness is demonstrated
though visualizing space curves in $3$-D space in Section~\ref{sec:exp}.
Another new contribution of the paper is that we propose effective heuristic strategies
to handle ``pseudo singular points'' in Section~\ref{sec:optimization}.
Last but not least, we provide a proof of Theorem~\ref{Theorem:robust}.

\section{A theoretical base algorithm}
\label{sec:theory}
It is highly nontrivial for continuation methods to guarantee that the polygonal chains
are $\epsilon$-close to the curve even when the curve contains no singular points.
Robust tracing without curve jumping must be involved.
In the literature, there are several techniques~\cite{Blum1997,Beltran2013,Martin2013} that can solve this.
Here we combine the technique of $\alpha$-theory~\cite{Blum1997,Beltran2013} with a technique developed in~\cite{WRF2017},
which has been used to estimate the error of numerically computed border curve in~\cite{Chen2016}.
In particular, we have Theorem~\ref{Theorem:robust},
which  provides a way to obtain $\epsilon$-approximation of a regular section of a curve.

\subsection{Preliminary}
We first recall some well-known results
which are important for the rest of the paper.

Throughout this paper, let $F=\{f_1,\ldots,f_{n-1}\}\subseteq \R[x_1,\ldots,x_n]$,
$B\subset \R^n$ be a bounded box and $\epsilon\in \R$ be a given precision.
Let $\Jac_{F}$ be the Jacobian of $F$, or simply $\Jac$ if no confusion arises.
We assume that $\dim{V_{\C}(F)}=\dim{V_{\R}(F)}=1$ and at almost all points of $V_{\C}(F)$,
$\Jac$ has full rank $n-1$.
Let $M_i$, $i=1,\ldots,{n}$, be respectively the submatrix of $\Jac$ by removing the $i$-th column of $\Jac$.
Let $\Delta_i$, $i=1,\ldots,n$, be the determinant of $M_i$.
Then the zero set of $F\cup\{\Delta_1,\ldots,\Delta_n\}$ in $\R^n$
is the set of singular points of $V_{\R}(F)$,
which is finite under our assumptions.
Denote by ${\sf SingularPoints}(F, B)$ an operation which computes
the set of {\bf singular points} of $V_{\R}(F)$ inside $B$.

Let $W$ be a closed smooth component of $V_{\R}(F)$.
Let $L:{a}{x}=d$ be a hyperplane in $\R^n$.
Let ${x}^*$ be any point of $W$.
The distance from $x^*$ to $L$ is given by $\displaystyle \delta({x}^*)=\frac{|{a}{x}^*-d|}{\|{a}\|_2}$.
By the method of Lagrange multipliers, the local extrema of the constrained optimization problem
${\rm min}~\delta({x})$ subject to $F({x})=0$ can be obtained by solving the system
$G := F\cup \{\Jac^t\lambda=a\}$, where $\lambda=(\lambda_1,\ldots,\lambda_n)$ are introduced
new variables.
Let $(x^0,\lambda^0)$ be a point of $V_{\R}(G)$ such that  $\Jac_F$ has full rank at $x^0$,
then by Proposition $2.2$ of~\cite{WR13}, for almost all choices of $a$ in $\R^n$,
$(x^0,\lambda^0)$ is an isolated point of $V_{\R}(G)$.
Therefore, we can apply homotopy continuation method to find such isolated points.
Denote by ${\sf WitnessPoints}(F, B)$ an operation which finds all such $x^0\in B$,
denoted by $S$.
If $x^0\in W$,  $\Jac_F$ has full rank at $x^0$,
thus $x^0\in S$.
In other words, $S$ contains at least one point (called {\bf witness point}) from each closed smooth component $W$ of $V_{\R}(F)$.
Note that it is possible that $S$ may also contain witness points of non-closed components of $V_{\R}(F)$.

Let $\Jac(p)=U\Sigma V^T$ be the singular value decomposition (SVD) of $\Jac$ at a regular point $p$ of $V_{\R}(F)$,
where $U$ and $V$ are orthogonal matrices of dimension $(n-1)\times (n-1)$ and $n\times n$
and $\Sigma={\rm diag}(\sigma_1,\ldots,\sigma_{n-1})\in\R^{(n-1)\times n}$, where $\sigma_1>\cdots>\sigma_{n-1}>0$.
The Eckhart-Young Theorem says that  $\displaystyle\sigma_{i}=\underset{{\rm rank}(B)=i-1}{{\rm min}}\|\Jac-B\|_2$, $i=1,\ldots,n-1$.
Therefore, the smallest singular value $\sigma_{n-1}$ measures the 2-norm distance of $\Jac$
to the set of all rank-deficient matrices.
In other words, it tells how $\Jac$ is close to be singular.

Let $V=[v_1\mid\dots\mid v_n]$, where each $v_i$ is a column vector.
By $\Jac(p)=U\Sigma V^T$, we have $\Jac(p) v_n=0$.
That is $v_n$ is a tangent vector which can be chosen as
the direction to trace the curve.
We define an operator ${\frak t}$ which computes $v_n$ from $\Jac(p)$,
that is $v_n={\frak t}(\Jac(p))$.
More precisely, to trace the curve, we adopt a {\bf predictor-corrector} method.
In the predictor step, we choose a step size $h>0$ and compute the predictor point
$q=p+hv_n$. In the corrector step, we apply Newton iterations 
$\displaystyle q_{i+1}=q_i-\bigg( \begin{smallmatrix} {\Jac(q_i)}  \\ {{\frak t}(\Jac(q_i))^T} \end{smallmatrix} \bigr)^{-1} \bigl( \begin{smallmatrix}  {F(q_i)} \\ 0 \end{smallmatrix} \bigg) $,
$i\geq 0$, with $q_0=q$, until a stopping criterion is met.

\begin{definition}
    Let $X$ and $Y$ be two non-empty subsets of a metric space $(M,d)$.
    The {\em Hausdorff distance} $d_H(X,Y)$ is defined by
    $$
    d_H(X,Y)={\rm max}\{ {\rm sup}_{x\in X}{\rm inf}_{y\in Y}d(x,y),  {\rm sup}_{y\in Y}{\rm inf}_{x\in X}d(x,y) \}.
    $$
    Given $F=\{f_1,\ldots,f_{n-1}\}\subseteq \R[x_1,\ldots,x_n]$ with $\dim{V_{\R}(F)}=1$, a finite box $B\subset \R^n$ and a given precision $\epsilon\in R$.
    A  set ${\cal S}$ of polygonal chains contained in $B$ is called an $\epsilon$-approximation of $V_{\R}(F)$ if $d_{H}(\cup_{P\in {\cal S}} P, V_{\R}(F)\cap B)\leq \epsilon$ holds.
  \end{definition}

\subsection{Robust tracing of regular curve}

  Let $D$ be the unit disk centered at the origin.
  Let $B$ be a bounding box of $\R^n$.
  W.l.o.g, we assume that $B\subset D$ and that $K(F)={\rm max}(\{\lVert \nabla \Jac_{ij}(z) \rVert)\mid z\in D\})\leq 1$ holds,
  which can always be achieved by shifting and rescaling.

  Let $\tilde{{z}_0}$ be an approximate point of $V_{\R}(F)$,
  such that there exists a $\tau$ to make the intersection
  of $\lVert z-\tilde{z_0}\rVert\leq\tau$ and $V_{\R}(F)$ have only
  one connected component
\footnote{This component is a subset of a connected component $C$ of $V_{\R}(F)$
and the point $\tilde{z_0}$ belongs to the Voronoi cell of $C$.} 
and the line in the gradient direction of $F$ at $\tilde{z_0}$
  have only one intersection point $z_0$ with the component, see Fig.~\ref{fig:branch}.
  We call $z_0$ the {\em associated exact point} of $\tilde{z_0}$ on $V_{\R}(F)$.

 \begin{figure}
  \centering
  \includegraphics[width=0.6\textwidth]{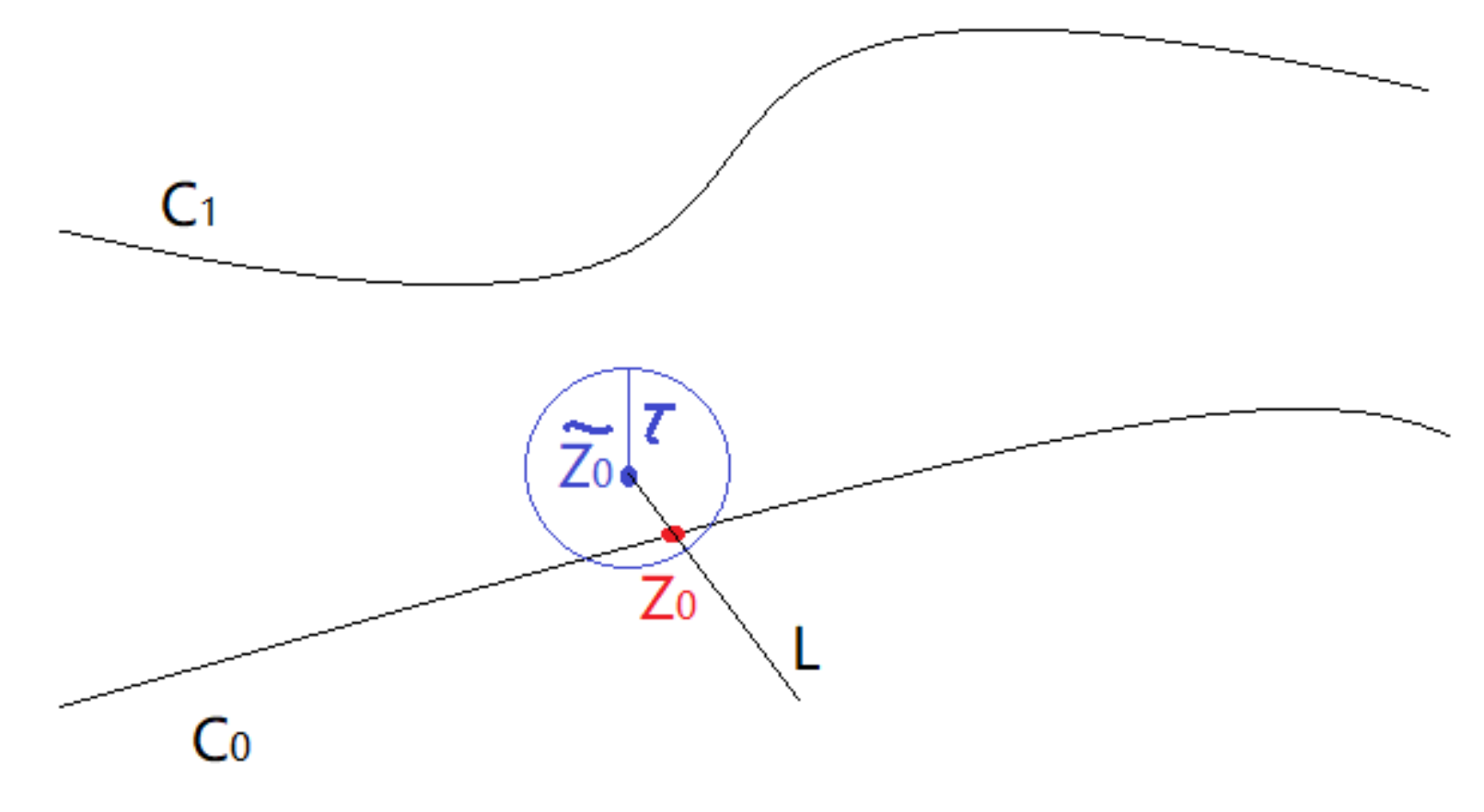}
  \caption{The associated exact point of an approximate point of the curve.}
  \label{fig:branch}
 \end{figure}

Next, we recall a result proved in \cite{WRF2017}  (Theorem 3.9).  

\begin{lemma}
  \label{Lemma:jump}
  Let $\sigma_0$ be the smallest singular value of $\Jac_F(z_0)$.
  Let $\rho\geq 1$ and $\omega(\rho):=\sqrt {2 \left( 2\,\rho-1 \right)  \left( 2\,\rho - 2\,\sqrt {\rho\, \left( \rho-1 \right) } -1 \right) }$.
  Note that $\omega(\rho)$ is a monotone decreasing function with limit $1$ as $\rho\rightarrow+\infty$.
 Assume that $2\rho>3\omega$ holds (which is true for any $\rho\geq 1.6$).
 Let $\mu=\sqrt{n(n-1)}$ and $\displaystyle s_0=\frac{\sigma_0}{2\mu\rho_0}$, $\rho_0\geq 1.6$.
 
As illustrated in Fig.~\ref{fig:error}, let $L_0$ be a hyperplane which is perpendicular to the tangent line  
of $V_{\R}(F)$ at $z_0$ of distance $s_0$ to $z_0$.
Let $z_1$ be a zero of $V_{\R}(F)\cap L_0$.
Then $z_1$ is on the same component with $z_0$ if and only if 
\begin{equation}
\label{eq:moveDist}
\|{z}_1 - {z}_0\|_2 < \omega(\rho_0) \cdot  s_0.
\end{equation}
\end{lemma}

  Now, let $L_1$ be a hyperplane which is perpendicular to the tangent line  
 of $V_{\R}(F)$ at $z_1$ of distance $s_1$ to $z_1$ such that $z_0$
 is a zero of $V_{\R}(F)\cap L_1$.
 Let $\sigma_1$ be the smallest singular value of $\Jac_F(z_1)$.
 If there exists a $\rho_1\geq 1.6$ such that $\displaystyle s_1=\frac{\sigma_1}{2\mu\rho_1}$,
 by Lemma~\ref{Lemma:jump}, $\|{z}_1 - {z}_0\|_2 < \omega(\rho_1) \cdot  s_1$ holds, 
 as $z_0$ is on the same component with $z_1$.
 If such $\rho_1$ does not exist, we can always increase the value of $\rho_0$,
 thus reduce the step size $s_0$ such that the condition $\rho_1\geq 1.6$ is satisfied.
 Indeed, by Weyl's theorem~\cite{Stewart90}, $|\sigma_1-\sigma_0|\leq \| \Jac(z_1)-\Jac(z_0)  \|_2$.
 By Lemma 1 of \cite{Chen2016}, $\| \Jac(z_1)-\Jac(z_0)  \|_2\leq \mu\|{z}_1 - {z}_0\|_2$.
 Thus  $|\sigma_1-\sigma_0|\leq \mu\|{z}_1 - {z}_0\|_2$.
 As $z_1$ approaches $z_0$, reducing $s_1$ implies increasing $\rho_1$.

 To sum up, one can find $\rho_i\geq 1.6$ such that $\displaystyle s_i=\frac{\sigma_i}{2\mu\rho_i}$
 and $\|{z}_1 - {z}_0\|_2 < \omega(\rho_i) \cdot  s_i$, for $i=0,1$.
 Let $\rho^*={\rm min}(\rho_0, \rho_1)\geq 1.6$.
 Let $h=\|{z}_1 - {z}_0\|_2$.
We define a cone with $z_0$ as the apex, the tangent line at $z_0$ as the axis, 
and the angle deviating from the axis being $\displaystyle \theta := \arccos\bigg(\frac{1}{\omega(\rho^*)}\bigg)$. 
By above analysis, the curve from $z_0$ to $z_1$ must be in this cone when the step size is small. 
Similarly, we can construct another cone with $z_1$ as the apex, 
the tangent line at $z_1$ as the axis, 
and the angle deviating from the axis being $\displaystyle \theta := \arccos\bigg(\frac{1}{\omega(\rho^*)}\bigg)$, 
such that it contains the curve from $z_1$ back to $z_0$.
Figure~\ref{fig:error} illustrates the two cones and the curve contained in them.


\begin{figure}
  \begin{center}
    \includegraphics[width=0.7\textwidth]{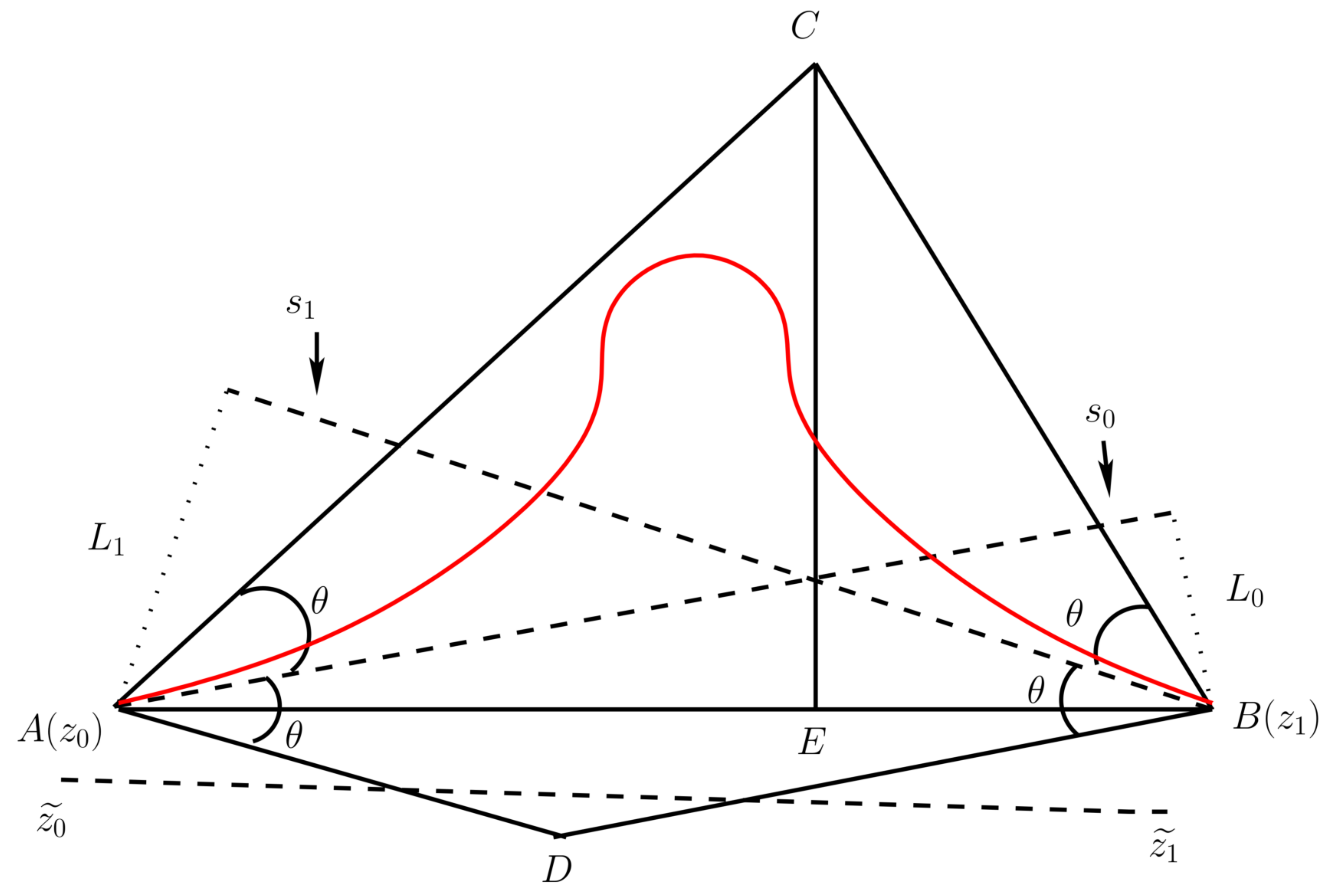}
\end{center}
\caption{A 2D image illustrating the intersection of two cones.}
\label{fig:error}
\end{figure}

From Figure~\ref{fig:error}, 
we know that $|CE|<|AE|\tan(2\theta)$ and $|CE|<|EB|\tan(2\theta)$ hold.
Since $|AE|+|EB|=h$, we deduce that $\displaystyle |CE|<\frac{h}{2}\tan(2\theta)=\frac{h}{2}\tan\bigg(2\arccos\bigg(1/\omega(\rho_*)\bigg)\bigg)$.

To summarize, we have the following lemma.
\begin{lemma}
  \label{Lemma:dist}
  Use notations in the above analysis,
  let ${\cal C}_{z_0z_1}$ be the curve segment between $z_0$ and $z_1$ in $V_\R(F)$.
  Let $\sigma={\rm max}(\sigma_0,\sigma_1)$.
  The  Hausdorff distance between $C_{z_0z_1}$ and the segment $\overline{{z}_0 {z}_1}$
  is at most $\displaystyle \frac{\omega(\rho_*)\sigma}{4\mu\rho_*}\tan\bigg(2\arccos\bigg(\frac{1}{\omega(\rho_*)}\bigg)\bigg)$.
\end{lemma}

In practice, $\rho_1$ and $z_1$ cannot be obtained exactly as before.
But one can always apply interval arithmetic to obtain an approximation $\tilde{z_1}$
such that the {\em associated exact point} of $\tilde{z_1}$, called $z_1$
satisfies Lemma~\ref{Lemma:dist}.

Assume that $\|z_1-\tilde{z_1}\|_2\leq\tau$.
As before, by Weyl's theorem~\cite{Stewart90} and Lemma 1 of \cite{Chen2016},
we have $\sigma_i\leq \mu\tau+\tilde{\sigma_i}$, $i=1,2$.
Let $\tilde{\sigma}={\rm max}(\tilde{\sigma_0},\tilde{\sigma_1})$.
Combining with Lemma~\ref{Lemma:dist}, we have the following result.

  \begin{theorem}
  \label{Theorem:robust}
  The  Hausdorff distance between $C_{z_0z_1}$ and the segment $\overline{\tilde{z}_0\tilde{z}_1}$
  is at most 
  $$
  \tan\bigg( 2\arccos\bigg(\frac{1}{\omega(\rho_*)}\bigg)  \bigg)\frac{\omega(\rho_*)}{4\mu\rho_*}(\mu\tau + \tilde{\sigma})+\tau,
  $$
  which is no greater than $\displaystyle 1.082\tau + 0.082\frac{\tilde{\sigma}}{\mu}$ since $\rho_*\geq 1.6$.
\end{theorem}

\subsection{Handling singularities}
It is a well known fact that tracing a curve near singularities is difficult,
as illustrated in Fig~\ref{fig:eye}.
\begin{figure}
\begin{minipage}[t]{0.43\linewidth}
\centering
\includegraphics[width=\textwidth]{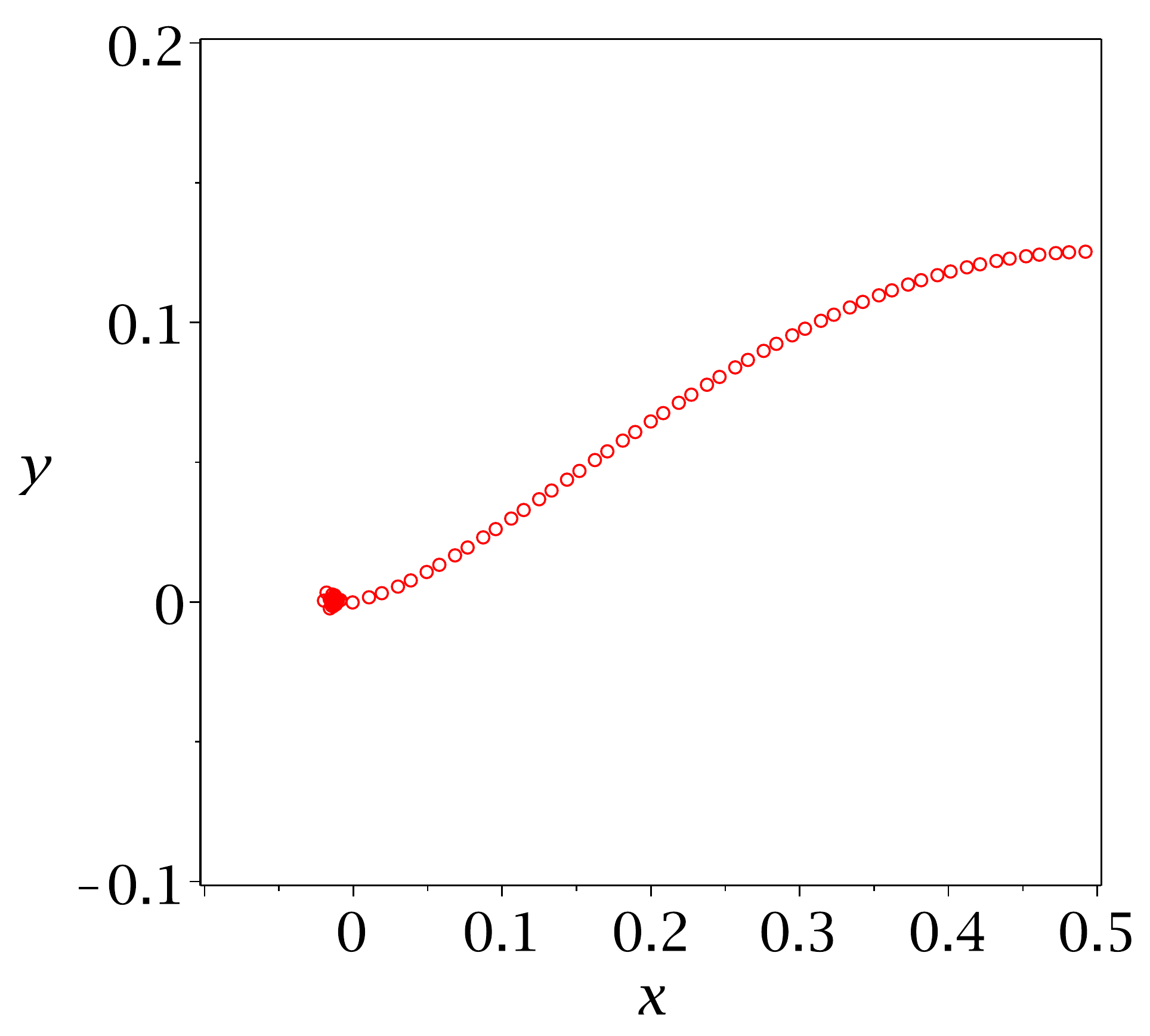}
\end{minipage}
\begin{minipage}[t]{0.45\linewidth}
\centering
\includegraphics[width=\textwidth]{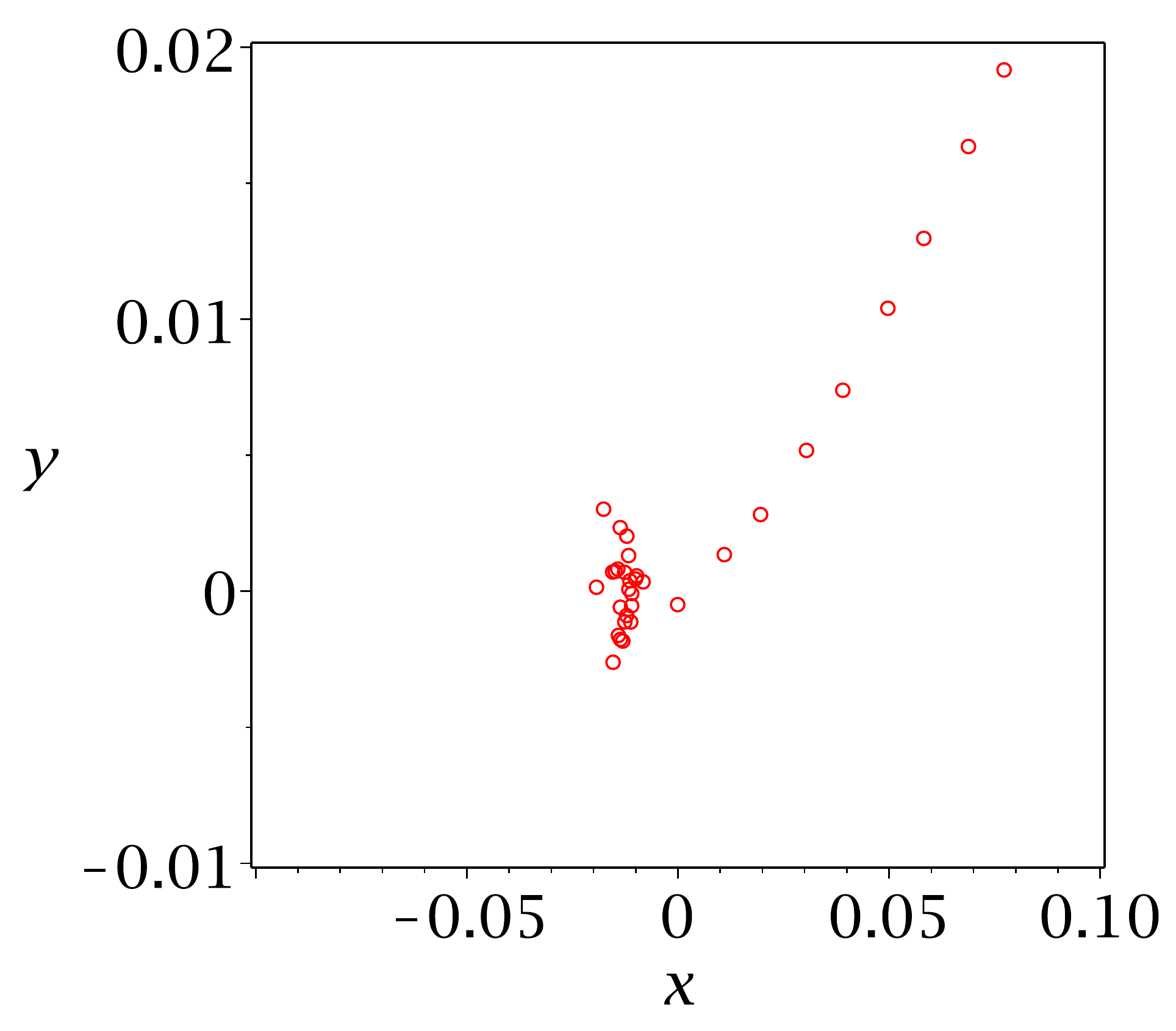}
\end{minipage}
\caption{Tracing the curve near a singular point.}
\label{fig:eye}
\end{figure}
The left subfigure illustrates tracing the zero of 
$f = {y}^{2}- \left( -{x}^{2}+x \right) ^{3}$ starting with a regular point,
where the right subfigure zooms in
the part of the left subfigure near the origin, which is a singular point.
We see that it may be difficult for curve tracing to escape out of the area near the origin,
as near the origin, Newton's method requires to solve
a linear system $Az=b$ with a very large condition number.
As a result, the errors are radically amplified.

Even worse, the topology of the curve near singularities is not numerically stable,
as illustrated by Example~\ref{ex:per}.
\begin{example}
 \label{ex:per}
  Let $f := x^2-y^2$. A slight perturbation of its coefficients
  changes completely the local topology near its singular point $(0,0)$,
  as depicted in Fig.~\ref{fig:per}.
\end{example}

\begin{figure}

  \begin{minipage}[t]{0.3\linewidth}
    \centering
    
\includegraphics[width=\textwidth]{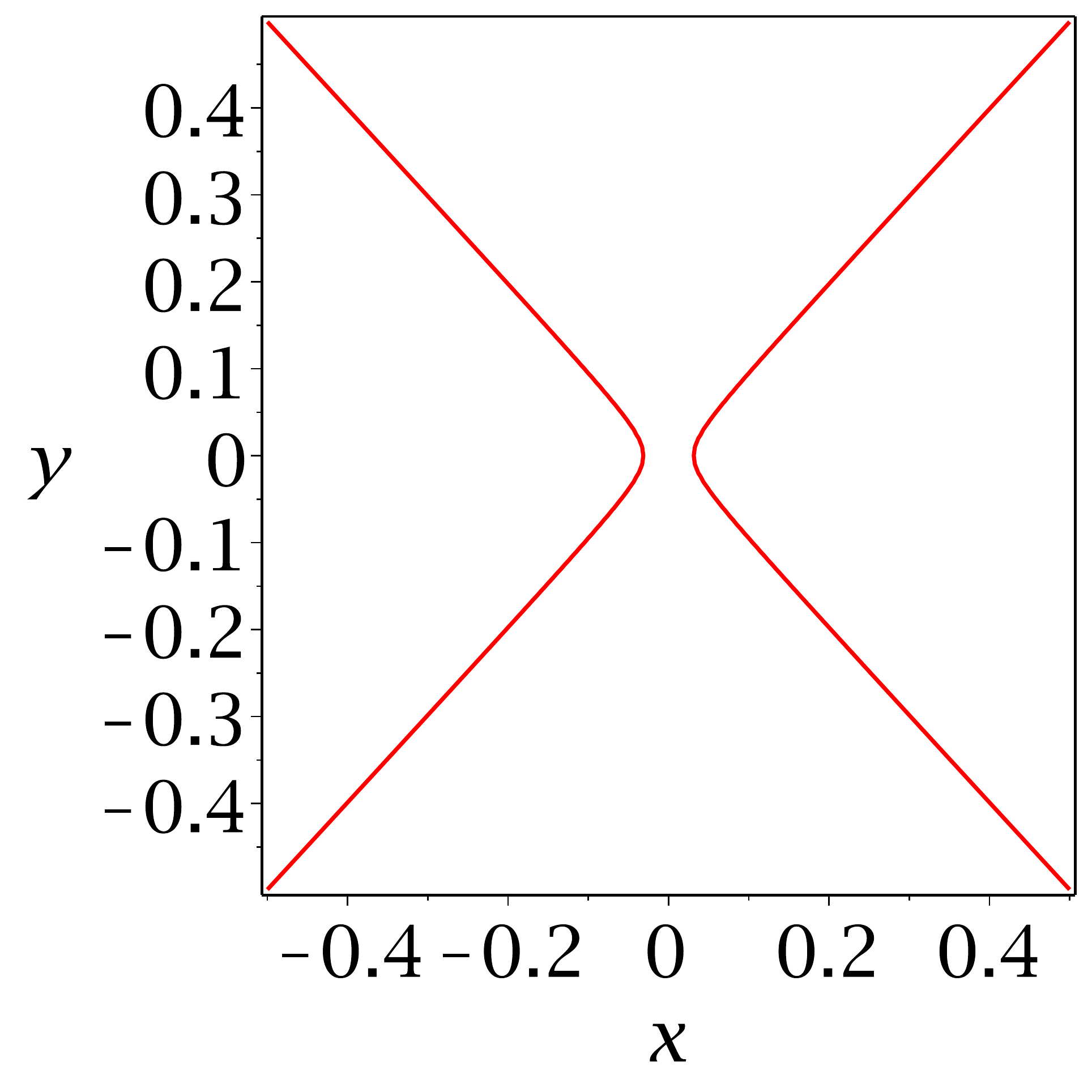}
{Plot of $f-0.001$.}
\end{minipage}
\begin{minipage}[t]{0.3\linewidth}
  \centering

\includegraphics[width=\textwidth]{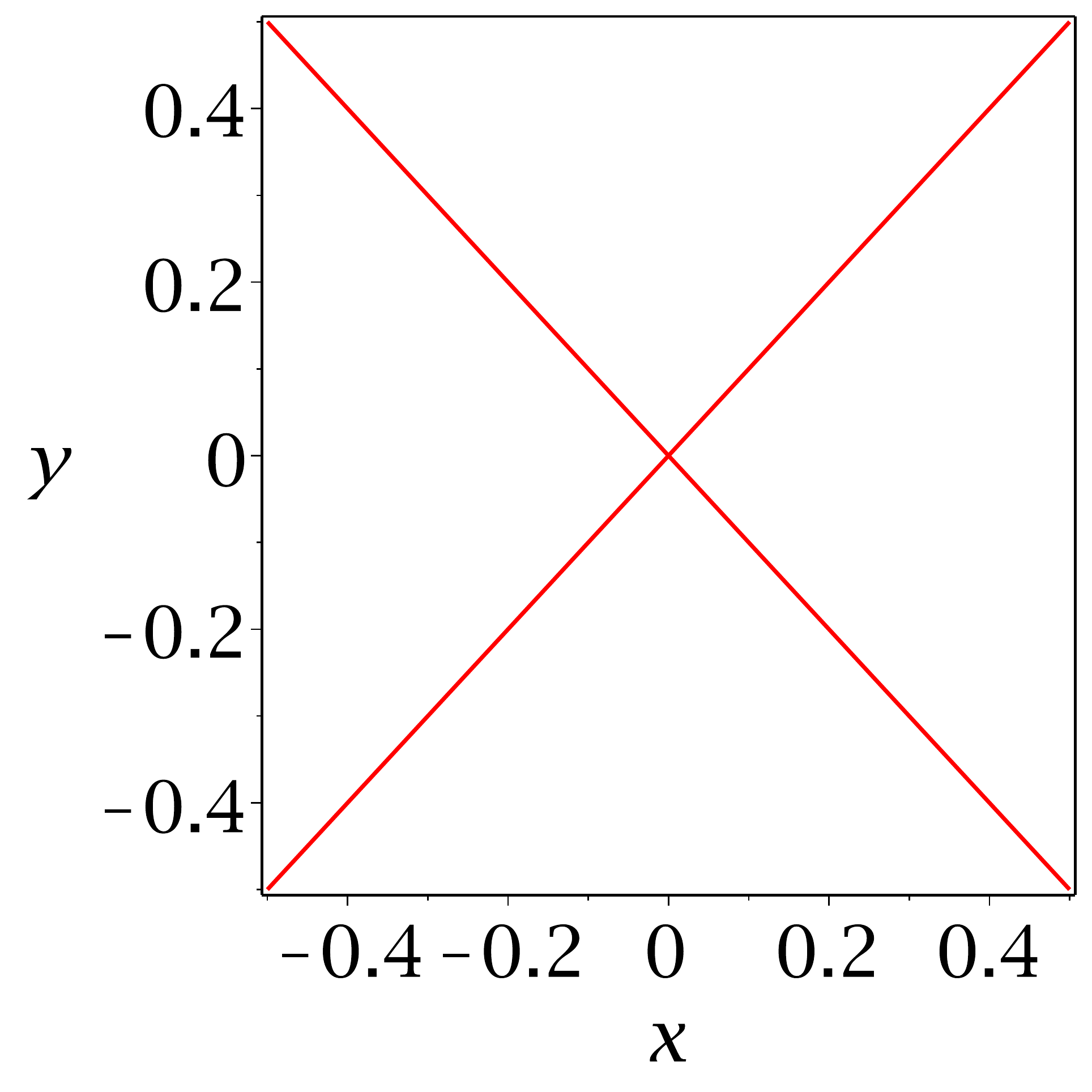}
  {Plot of $f$.}
\end{minipage}
\begin{minipage}[t]{0.3\linewidth}
  \centering

\includegraphics[width=\textwidth]{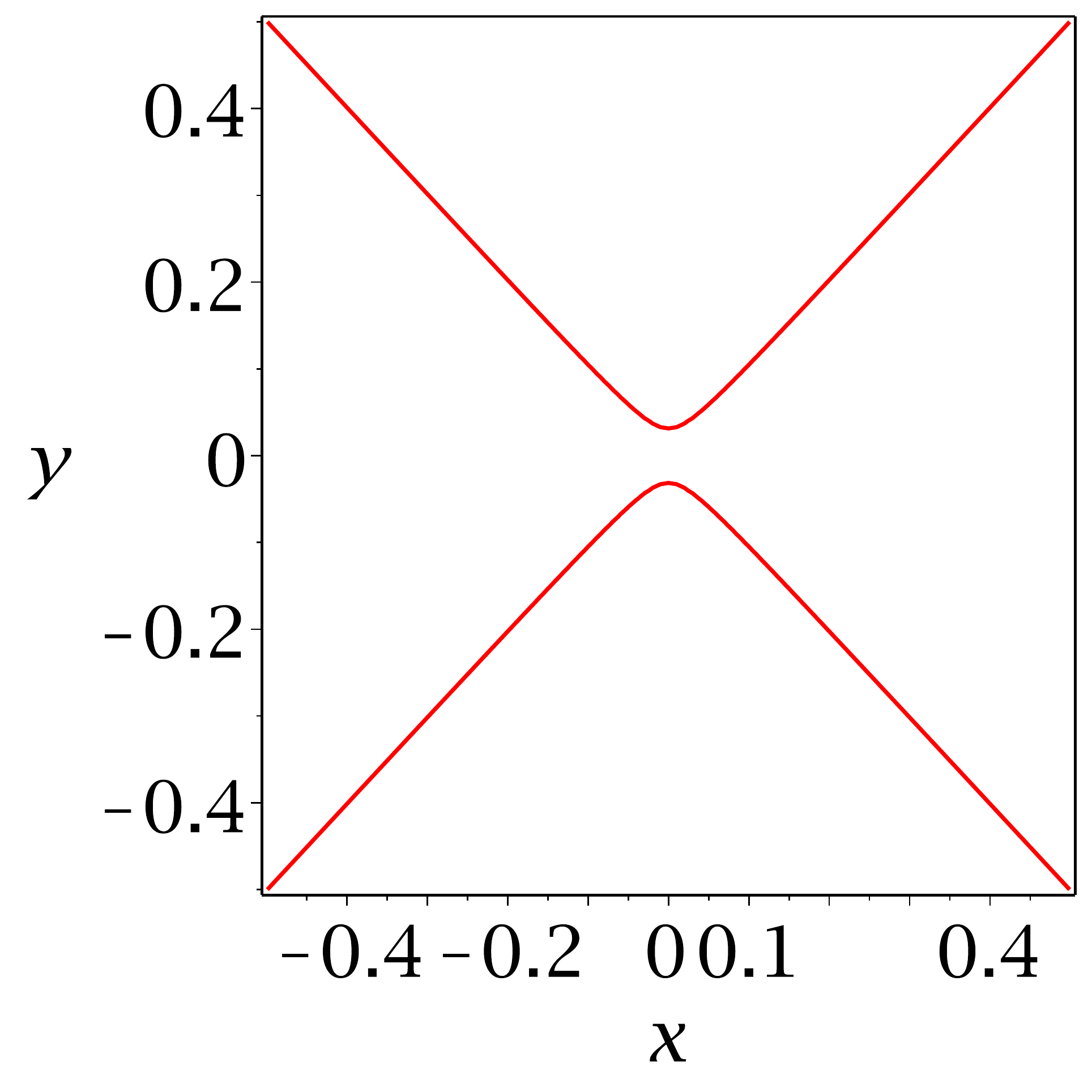}
  {Plot of $f+0.001$.}
\end{minipage}
 \caption{Plot of $f$ and its perturbations.}
\label{fig:per}
\end{figure}

So instead of tracing through a singular point,
we bypass it. 
Before presenting an algorithm, we first use a simple example to illustrate the main idea.
\begin{example}
  Consider the polynomial
  $
  f := 6\,x{y}^{7}+85\,{x}^{4}{y}^{3}-60\,{x}^{2}{y}^{5}-32\,{x}^{2}{y}^{3}+14\,{x}^{4}-35\,{y}^{4}.
  $
  Its real zero set is displayed in Fig.~\ref{fig:intro} as the red curve.
 \end{example} 
  It has three connected components inside the box $-3\leq x\leq 3, -4\leq y\leq 2$.
  The component on the top has an isolated singular point $(0,0)$, colored in green.
  To plot this component, we first draw a circle centered at the origin, which
  has four intersection points with the curve, colored in black.
  Then we trace the four branches starting with the four black points until meeting the boundary.
  Next we plot the component at the left bottom corner. To do that, we start with a blue point,
  which is an intersection point of the curve with a boundary of the box, and trace the curve until
  meeting a boundary of the box.
  At last, we plot the closed component at the right bottom corner.
  To do that, we compute critical points of the curve in $x$-direction and get two yellow points.
  Starting with any point of them, trace the curve until meeting the point itself.
  Finally we plot the singular point.
  See the right subfigure of Fig.~\ref{fig:intro} for a visualization of the approximation.
  \begin{figure}
\begin{minipage}[t]{0.43\linewidth}
\centering
\includegraphics[width=\textwidth]{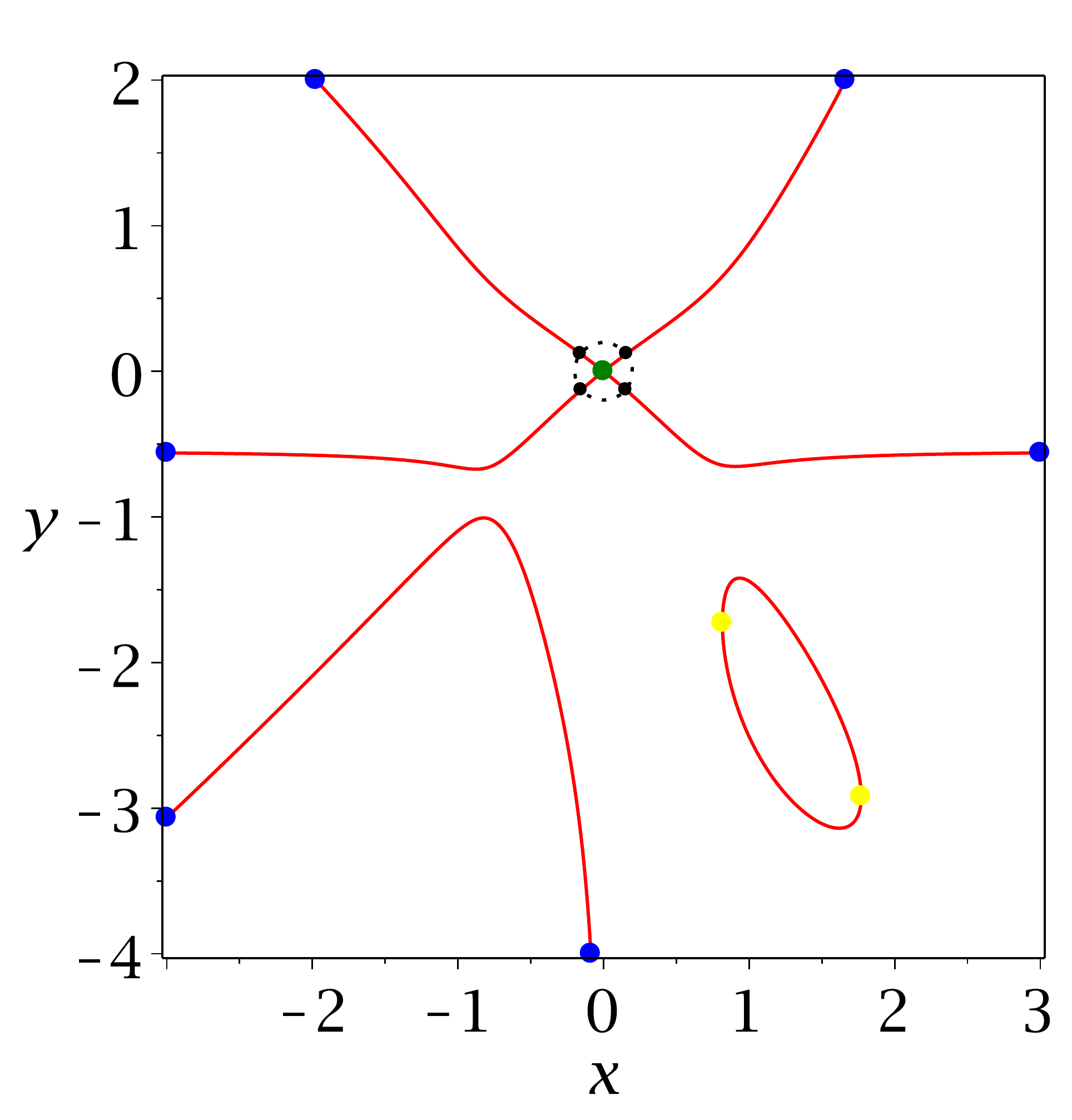}
\end{minipage}
\begin{minipage}[t]{0.45\linewidth}
\centering
\includegraphics[width=\textwidth]{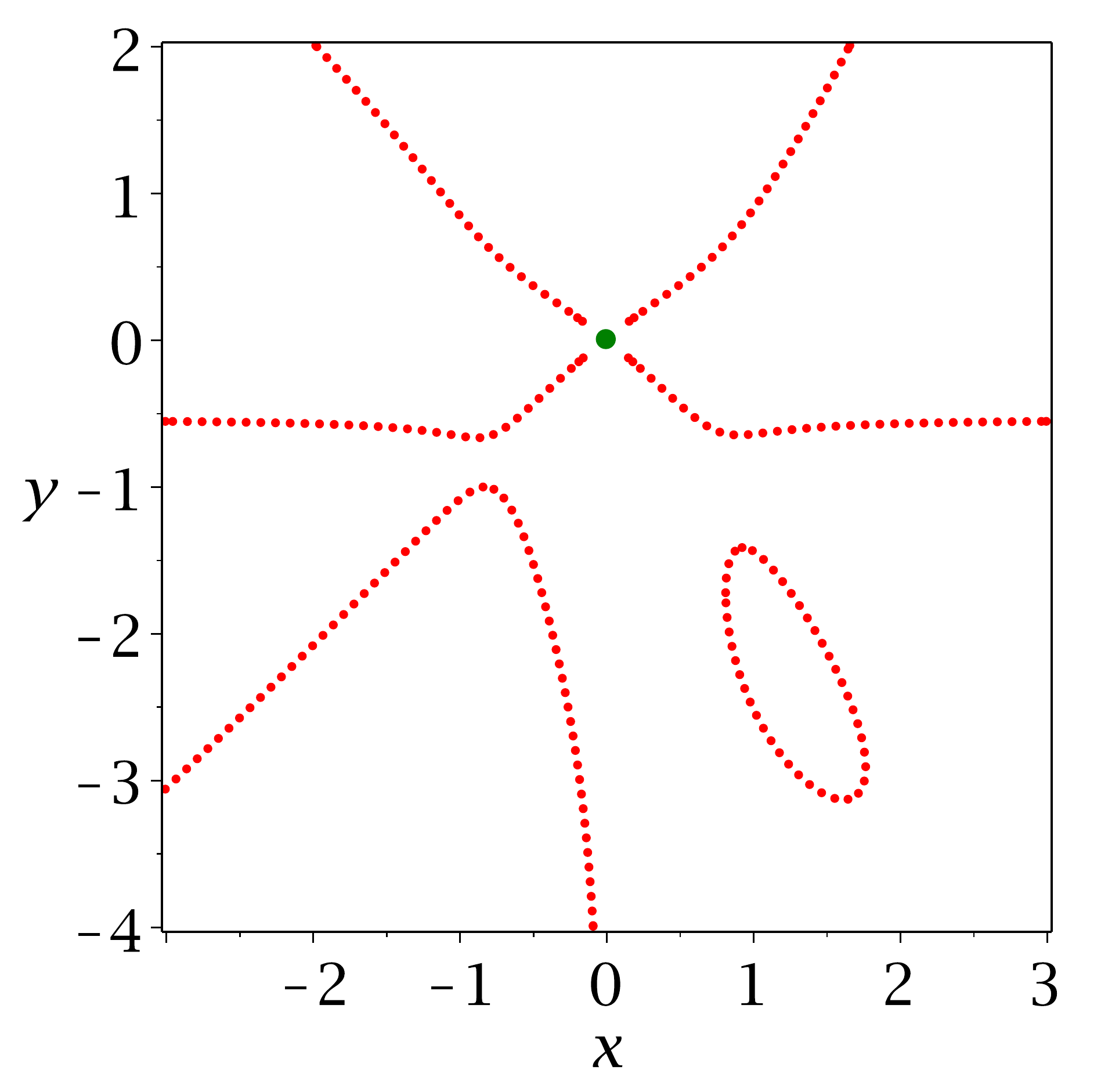}
\end{minipage}
\caption{{Left: the curve and key points. Right: an approximation of the curve ($\epsilon=0.4$).}}
\label{fig:intro}
\end{figure} 
  Note that the above procedure plots the whole curve
  inside the box except the part in a small neighborhood
  of the origin, which is simply replaced by a point.
   Such an approximation is numerically more stable than describing exactly
  the topology near the origin, as illustrated by Example~\ref{ex:per}.
  Moreover, in applications such as solving parametric polynomial systems, the curve
  is a border curve and such an approximation suffices to
  answer exactly the number of real solutions of the parametric system
  in an open cell of the complement of the curve.

  \begin{itemizeshort}
    \item Algorithm {\bf ApproxPlotBase}
\item Input: a finite set of polynomials $F=\{f_1,\ldots,f_{n-1}\}\in\Q[x_1,\ldots,x_n]$; a bounding box $B\subset \R^n$, and a given precision $\epsilon$.
  \item Output: an $\epsilon$-approximation of $V_{\R}(F)$.
  \item Assumptions:
    $(i)$ the singular points are not on the boundary of the box $B$;
    $(ii)$ the distance between two singular points is at least $\epsilon$;
  \end{itemizeshort}
\begin{enumerateshort} 
  \item Compute the singular points $S_0={\sf SingularPoints}(F, B)$.
\item Compute the intersection of the curve with spheres centered at the singular points with radius less than $\epsilon/2$
      \footnote{One could replace the spheres with axis aligned boxes inside them.}.
      Set $S_1$ to be the set of these points, called fencing points (around singular points).
      We denote by ${\sf FencingPoints}$ an operation to compute such points.
      Let $\Delta$ be the set of balls associated with these spheres.
 \item Compute the intersection of the curve with the boundaries.
      Set $S_2$ to be the set of these points.
    \item Compute the witness points of $V_{R}(F)$ (inside $B$) $S_3:={\sf WitnessPoints}(F, B)$.
  Remove from $S_3$ points that are already inside any balls in $\Delta$.
\item Starting with a point in $S_1$, trace the curve robustly based on Theorem~\ref{Theorem:robust} until meeting ($\epsilon$-close to) a point in $S_1$ or $S_2$.
      Remove the starting point and the corresponding points met in $S_1$ or $S_2$. Repeat Step $(5)$ until $S_1=\emptyset$.
      Let the resulting set of polygonal chains be $P_1$.          
    \item If $S_2\neq\emptyset$, starting with a point in $S_2$, trace the curve robustly until meeting a point in $S_2$.
          Remove the point met in $S_2$.
          Repeat Step $(6)$ until $S_2=\emptyset$.
          Let the resulting set of polygonal chains be $P_2$.          
        \item Remove points of $S_3$ which are already on the computed curve.
        \item If $S_3\neq\emptyset$, starting with a point in $S_3$, trace the curve robustly until closed curves are found.
              Remove point met during the tracing from $S_3$. Repeat Step $(8)$ until $S_3=\emptyset$.
          Let the resulting set of polygonal chains be $P_3$.
        \item Return $S_0\cup P_1\cup P_2\cup P_3$.
\end{enumerateshort}
\begin{remark}
  Assumption $(i)$ can relaxed by slightly shrinking or expanding the box.
  Assumption $(ii)$ can be relaxed by grouping the singular points into clusters.
  See Section~\ref{sec:optimization} for details.
\end{remark}

\begin{theorem}
  One can control errors of staring points and prediction-correction in the above tracing algorithm,
  such that  Algorithm {\bf ApproxPlotBase} computes an $\epsilon$-approximation of $V_{\R}(F)$.
\end{theorem}
\begin{proof}
  We remark that to obtain an $\epsilon$-approximation of the curve, one must have one witness point from each connected component of the curve.
  If a component is a solitary point, it must be in $S_0$.
  For the other components which intersect with the boundary or have singular points, the starting points are in $S_2$ and $S_1$
  respectively. Note that although $S_3$ may not contain witness points for every connected component of $V_{\R}(F)$,
  it must contain at least one witness points for each smooth closed component of $V_{\R}(F)$.
  By the assumptions, the polynomial systems with zero sets $S_i$, $i=0,\ldots,3$ are all zero-dimensional.
  If the interval Newton method~\cite{Shen2012} converges,
  the error of solving these zero-dimensional systems and the error of Newton iterations (in the corrector step),
  as well as the distance between the curve and the polygonal chains can be controlled to be much less than $\epsilon$ by Theorem~\ref{Theorem:robust}.
  Otherwise, one can reduce the step size until the $\alpha$-theory~\cite{Blum1997,Beltran2013}  guarantees the convergence of Newton iterations.
  Moreover, by Theorem~\ref{Theorem:robust}, curve jumping can be avoided.
  Finally note that in the $\epsilon/2$-neighborhood of the singular points, 
  the distance between the curve and the polygonal chains are less than $\epsilon$.
  Thus,  an $\epsilon$-approximation of $V_{\R}(f)$ can be computed.
\end{proof}

\section{Improvements}
\label{sec:optimization}
In this section, we propose several strategies
for improving the numerical stability of the tracing
algorithm in last section.
These improvements lead to a practical algorithm
presented in next section.
\subsection{Choice of tracing direction}
This first strategy is plotting the curve in the direction
away from the singular points rather than towards the singular point.
In practice, the former can better avoid curve jumping, as illustrated
by Fig.~\ref{fig:jump}.
\begin{figure}
\centering
\includegraphics[width=0.5\textwidth]{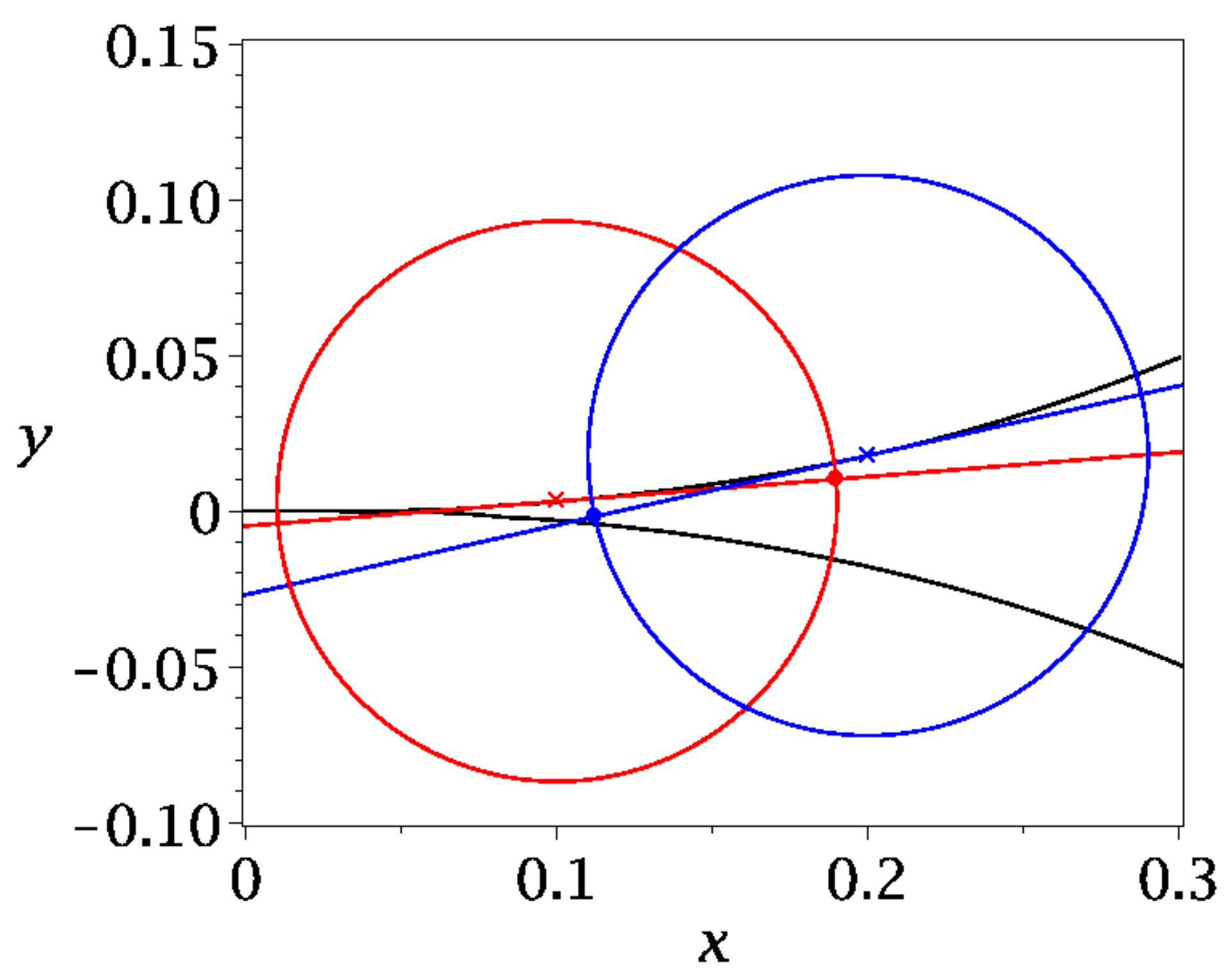}
\caption{Jump is more likely to happen when tracing towards singular points.}
\label{fig:jump}
\end{figure}
In this figure, the black curve is the locus of $f := x^5-y^2.$
To trace the upper branch, we have two possible starting points, namely the red $\times$ point, say $z_0$, and the blue $\times$ point, say $z_1$.
If we start from $z_0$ and move in the tangent direction towards $z_1$ in step size $0.09$,
we get a red $\bullet$ point close to the upper branch,
with which as an initial point, Newton iteration converges to a point still in the upper branch.
However, if we start from $z_1$ and move in the tangent direction towards $z_0$ in step size $0.09$,
we get a blue $\bullet$ point close to the lower branch.
As a result, Newton iteration converges to a point in the lower branch.

This justifies why we first start with fencing points around singular points instead of the boundary points to trace the curve,
as shown by Line $20$ of Algorithm~\ref{Algo:Plot}.
However, this first strategy does not consider the situation that there are two singular points in the same component,
for which  a try-stop-resume strategy is needed, as illustrated by examples in next subsection.

\subsection{A try-stop-resume tracing strategy}
\begin{example}
  Consider again the polynomial $f := {y}^{2}- \left( -{x}^{2}+x \right) ^{3}.$
  It is a closed curve with two singular points $(0,0)$ and $(1,0)$.
\end{example}
In Fig.~\ref{fig:try}, the algorithm first plots the red points starting from two fencing points
near $(0, 0)$
and stops when the singular values drop (at the two $\times$ points, which
are called front points).
See also line $37$ of Algorithm~\ref{Algo:main} for an implementation.
It then starts from the two fencing points near $(1,0)$ and plots the blue points,
which happen to meet the front points before singular values drop.
Checking if front points are met is implemented in Algorithm~\ref{Algo:main} from line $32$ to $35$.
\begin{figure}
\centering
\includegraphics[width=0.5\textwidth]{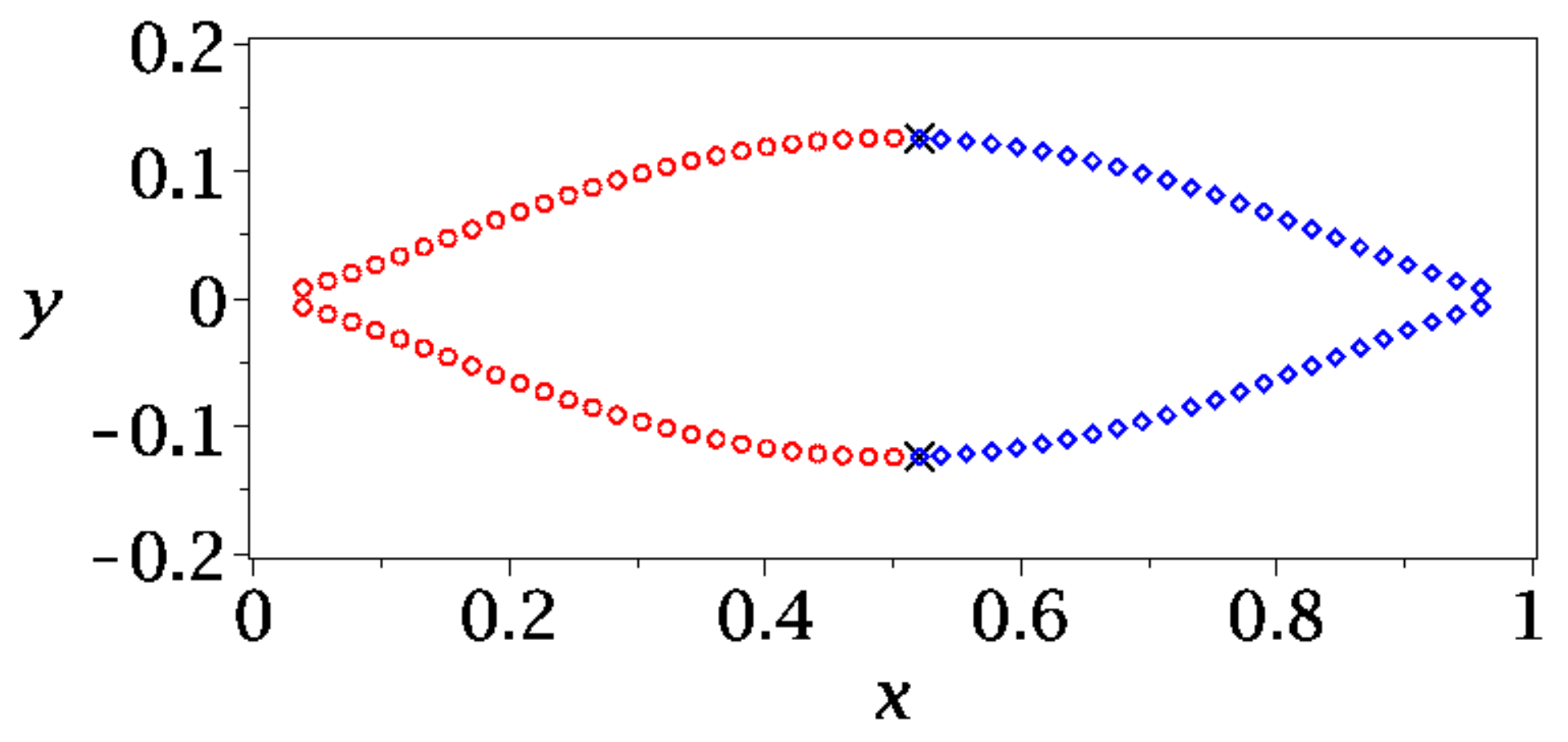}
\caption{{Try to plot the curve away from the singular points and stop when singular values drop.}}
\label{fig:try}
\end{figure}

The above example does not need the resuming step. 
Consider another one.
\begin{example}
  Consider
  {\small
  $$
  \begin{array}{rl}
    f := &-3375\,{y}^{14}-4050\,{x}^{4}{y}^{9}+108\,{y}^{13}-1215\,{x}^{8}{y}^{4}-648\,{x}^{2}{y}^{9}\\
    &+2700\,{y}^{11}+1620\,{x}^{4}{y}^{6}+1296\,{x}^{4}{y}^{5}-5400\,{x}^{2}{y}^{7}-3240\,{x}^{6}{y}^{2}\\
    &-1170\,{y}^{8}-864\,{x}^{6}y-810\,{x}^{4}{y}^{3}-720\,{x}^{2}{y}^{4}+4000\,{y}^{6}\\
    &+2400\,{x}^{4}y+540\,{y}^{5}+720\,{x}^{4}-1080\,{x}^{2}y-135\,{y}^{2}+800.
  \end{array}
  $$
  }
\end{example}
The locus of $f$ is visualized in Fig.~\ref{fig:resume}.
During the try phase, the algorithm starts with the fencing points at the bottom and
plots the red point. After all red parts have been plotted, it resumes and plots the blue parts and finally
the green parts. In this way, it avoids  directly tracing from the left singular point to the right one.
The resuming step is  implemented at line $21$ in Algorithm~\ref{Algo:Plot}, which calls
Algorithm~\ref{Algo:main} with the value of first argument $cwp$ replaced by front points ($front$).
\begin{figure}
\centering
\includegraphics[width=0.45\textwidth]{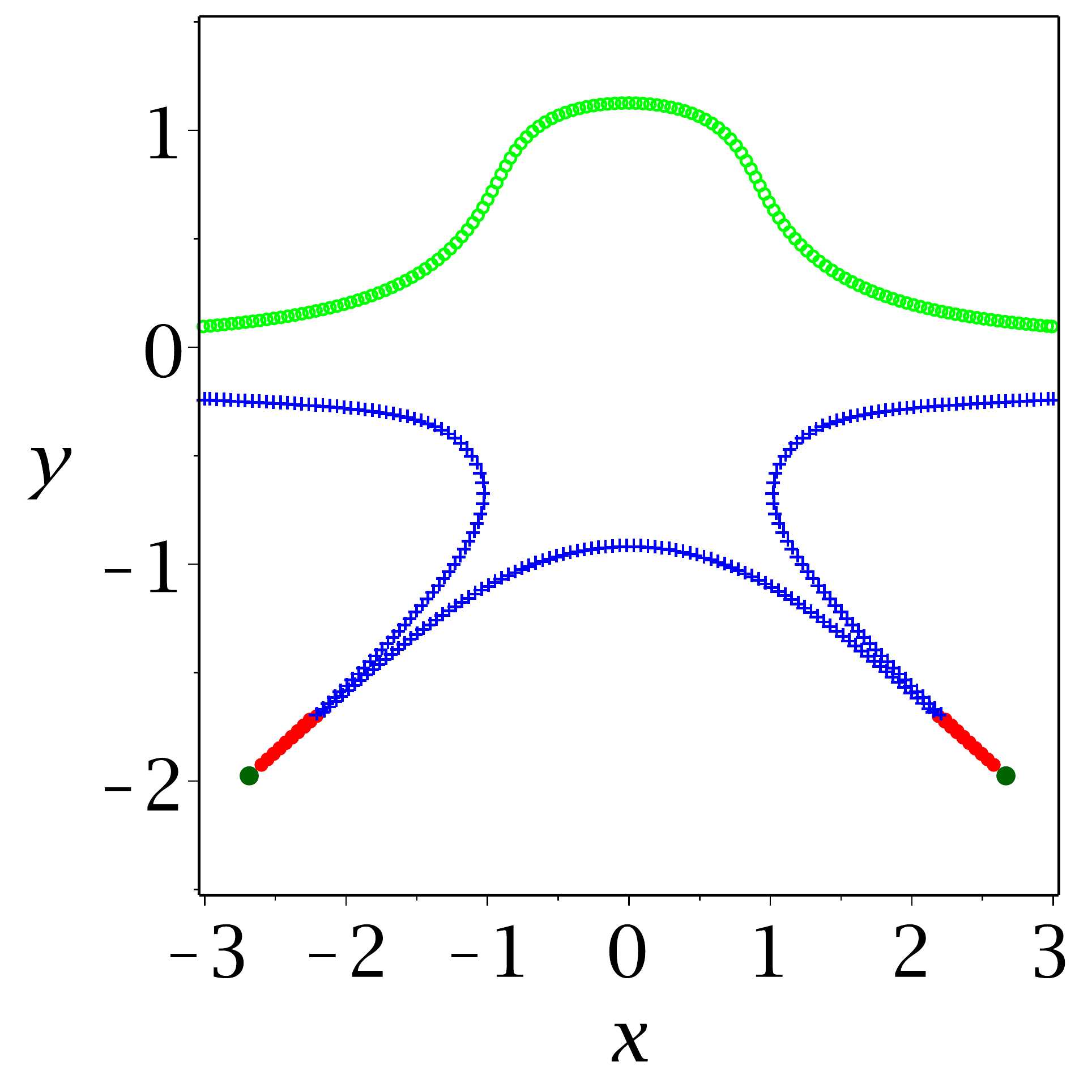}
\caption{{Try to plot the curve away from the singular points and stop when singular values drop and resume.}}
\label{fig:resume}
\end{figure}

\subsection{Handling clustered singular points}

The third improvement is to take clustered singular points into consideration.
We borrow the notion of natural cluster from~\cite{Bennett2016} on Voronoi vertices.
Given a set $S$ of singular points of $Z_{\R}(f)$ in a bounding box $B$.
For any disk $D(z,r)$ centered at $z$ of radius $r$,
let $\Delta_S(z,r)$ be the set of points in $S$ contained in $D(z,r)$.
If it is not empty, we call it a cluster of $S$.
It is called a natural cluster if $D(z,r)$ and $D(z,3r)$ contains exactly
the same set of points of $S$.
We call $D(z,r)$ an associated disk of $\Delta_S(z,r)$.
Note that the associated disks of two different natural clusters are disjoint
and the distance between their centers are at least $3r$.
For a given $S$, it is easy to generate a set of disjoint natural clusters
and their associated disks.
For instance, one can first sort the singular points in an ascending order
by the minimal distances from the point to the other points.
One can then check if the points form natural clusters of radius $r$
incrementally.
If not, we reduce the radius $r$ by half and repeat the above procedure.
Let $d$ be the minimal distances among points in $S$.
One can always obtain natural clusters of radius less than $d/3$.
Let ${\sf NatualClusters}$ be such an operation, which takes $S$ and $\epsilon/2$
as input and return a set ${\cal C}$ of natural clusters of radius $\delta$.
It is called at line $6$ of Algorithm~\ref{Algo:Plot}.
Let's consider an example.
\begin{example}
  \label{ex:discrim3}
  Let
  $
  g :=  -28\,{x}^{4}yz+58\,x{y}^{5}-65\,x{y}^{2}{z}^{3}+23\,{x}^{4}y+24\,{x}^{3}yz-64\,{x}^{2}{z}^{3}
        -32\,xy{z}^{3}-72\,x{y}^{2}z+6\,{z}^{4}+56\,xyz+1
        $
  and $f$ be the discriminant of $g$ w.r.t. $z$, which is an irreducible polynomial in $\Q[x,y]$.
  A visualization of it in the box $-1\leq x\leq 1,-1\leq y\leq 1$ is depicted in Fig.~\ref{fig:discrim3}.
  The two points $(-0.9257645305e-1, 0.7100519895)$ and $(-0.6009009066e-1, 0.7790657631)$
  on the top of Fig.~\ref{fig:discrim3} form a natural cluster of radius $0.1$.
  Note that near the left bottom corner of the box, the curve is plotted with a larger step size
  than the other parts.
  
  A second example for natural clusters is illustrated by Fig.~\ref{fig:cluster} in Section~\ref{sec:exp}.
\end{example}
  \begin{figure}
  \centering
  \includegraphics[width=0.45\textwidth]{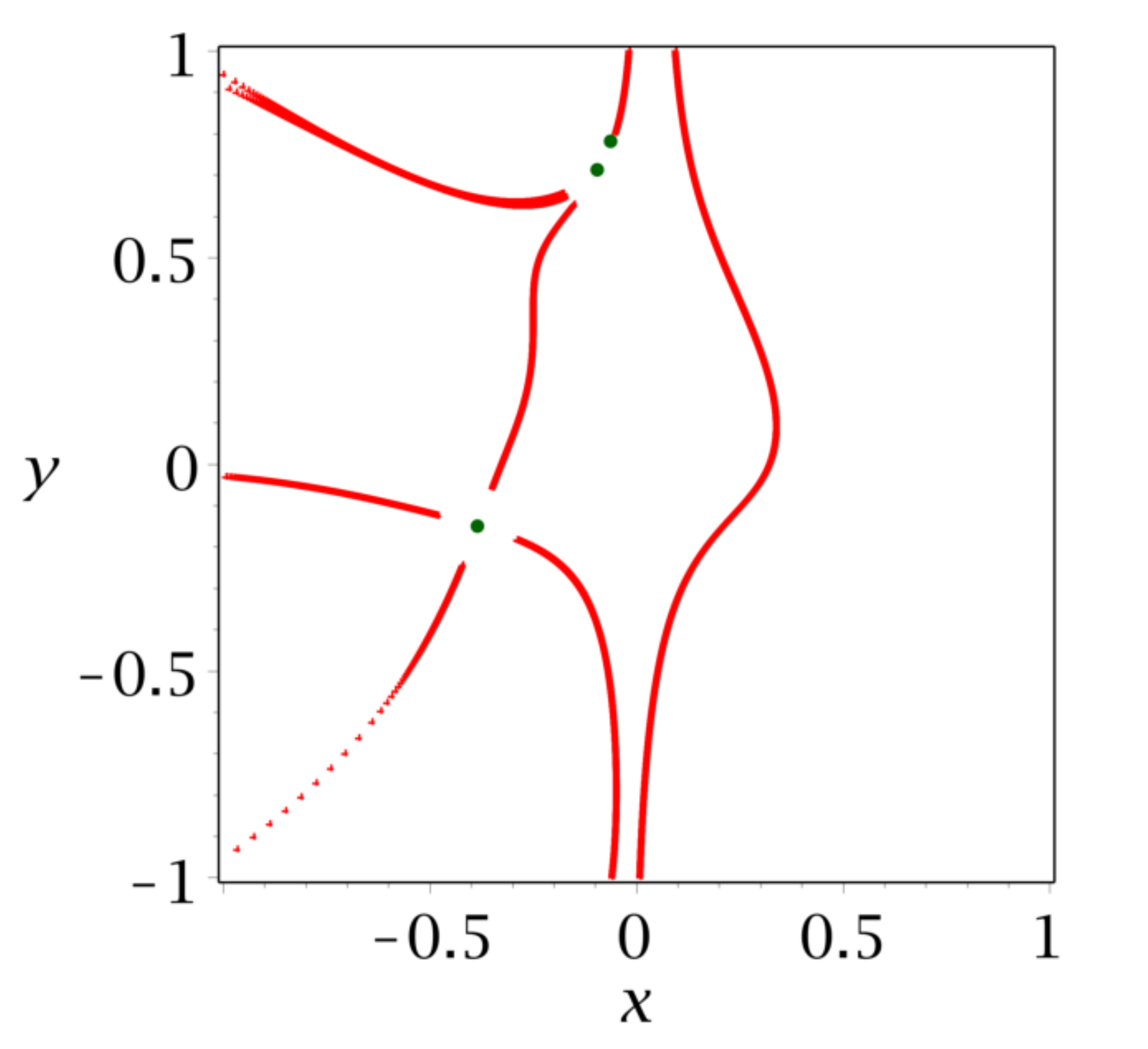}
     \caption{Plotting the curve with the help of natural clusters.}
   \label{fig:discrim3}
  \end{figure}

\subsection{Computing singular points under perturbation}
  
  The fourth improvement is to take into account the fact that sometimes the coefficients
  of input polynomials may be given approximately.
  As a result, as illustrated by an example earlier, an exact computation of singular points may be impossible.
  Consequently, we may fail to find Jacobian numerically
  singular parts of the curve if we directly compute the singular points.
  Thus, in the following, we propose a method to numerically computing singular points.

  We start by considering the case of plane curve defined by a polynomial $f\in\R[x_1,x_2]$.
  We are interested in computing an $\epsilon$-approximation of $V_{\R}(f)$.
  Suppose that we are given a $\tau\leq\epsilon$ perturbation to $f$ and let $g=f-\tau$ be the
  perturbed polynomial.
  Assume that $V_{\R}(g)$ is nonempty and the Hausdorff distance between $V_{\R}(g)$ and $V_{\R}(f)$
  is $\ll \epsilon$.
  
  According to Sard's theorem, for almost all such $\tau$,  any point of $V_{\R}(g)$
  is a regular point of the map $f$, which implies that $V_{\R}(g)$ is  a smooth manifold
  of dimension one (we assumed that $V_{\R}(g)$ is nonempty).

  Let $S :=\{q\mid q \mbox{~is a critical point of~} g\mbox{~and~} |g(q)|\leq \epsilon \}$.
  Let $p$ be a singular point of $f$.
  We have $|g(p)|=|f(p)-\tau|=|\tau|\leq\epsilon$, which implies that $p\in S$.
  On the other hand, since $d_H(V_{R}(f), V_{\R}(g))\ll\epsilon$, we have
  $d_H(p, V_{\R}(g))\ll\epsilon$.
  Thus from a given polynomial $g$, we get an approximation point $p$ of $g=0$
  ($d_H(p, V_{\R}(g))\ll\epsilon$),
  which turns to be a true singular point of its nearby polynomial $f=g+\tau$.
  We call such point a  ``pseudo singular point'' of $V_{\R}(g)$.
  Note that the set $S$ contains all such points (and possibly others).

  To summarize, if the constant coefficient of $f$ is slightly perturbed, in practice
  we can compute $S$ and use the point of $S$ to compute fencing points.
  
  In general, let  $F=\{f_1,\ldots,f_{n-1}\}$ be the set of polynomials defining the curve.
  Let $G=\{g_1,\ldots,g_{n-1}\}$ be a small perturbation of $F$.
  Finding ``pseudo singular points'' of $V_{\R}(G)$ is a difficult problem.
  Instead, we propose the following heuristic strategy to compute $S$, which works well in practice.
  
  Let $J$ be the Jacobian matrix of $G$.
  Let $J_i$, $i=1,\ldots,n$ be the submatrix of $J$ by deleting $i$-th column.
  Let $\Delta_i$ be the determinant of $J_i$.
  Then the critical points of $G:\R^{n}\rightarrow\R^{n-1}$ are the zeros of the system $\Delta := \{\Delta_1,\ldots,\Delta_n\}$,
  whose dimension in general is expected to be $n-2$.
  Let $E_i := \{g_1,\ldots,g_{i-1}, g_{i+1}, g_{n-1}\}\cup\Delta$.
  Let $S:=\cup_{i=1}^{n-1}\{q\mid q\in V_{\R}(E_i)\mbox{~and~} |g_{i}(q)|\leq \epsilon \}$.
  We denote by ${\sf PseudoSingularPoints}$ an operation to compute such $S$,
  which is called by Algorithm~\ref{Algo:Plot} at line $5$.
  It is possible that $S$ has identical points or points very close, which
  can be resolved by the natural cluster technique presented in last subsection.

\section{A practical algorithm}
\label{sec:algo}
In Section~\ref{sec:theory}, we presented a theoretical algorithm
to compute an $\epsilon$-approximation of a curve,
which may not be practical due to the small step size chosen.
In practice, one has to make a compromise between efficiency and accuracy.
Based on the improvement strategies in last section,
next we develop a more practical algorithm.
Instead of preventing curve jumping, in the algorithms below, we maintain
a simple data structure to record if a start point has been visited.
If a point is visited more than once, then there is a possible curve jumping.

  \begin{remark}
  The main features of Algorithm~{\sf ApproxPlot},
  such as tracing the curve away from the singular points,
  and grouping the singular points into natural clusters and the try-and-resume
  strategy has been explained in last section.
  Another feature of the algorithm is to detect curve jumping by counting
  the number of times that a fencing point or boundary point is visited.

  To achieve this, each fencing point, boundary point, or new front point
  generated due to the drop of singular value, is treated as an object
  with four attributes $(q, v, s, c)$, where $q$ is the point itself,
  $v$ is the tracing direction, $s$ is the singular value of $J_F(q)$
  and $c$ counts the times that $q$ is visited.
  For an object $ob$, the notation $ob.q$ means taking the value of the attribute $q$.
  Each $q$ should be visited one and only one time. If its visiting time $c>1$,
  there is a possible curve jumping at $q$.
  It is easy to check that if there is no curve jumping, after executing Algorithm~{\sf PlotMain},
  the value of any $c$ (counting visiting times of a fencing point or boundary point) can not be greater than $1$.
  Moreover, if the numerical errors are well controlled,
  after executing line $22$ of Algorithm~{\sf ApproxPlot},
  all the points in $rwp$ will only be on the closed components of the curve.
  Thus the value of any $c$ can not increase after executing Algorithm~{\sf PlotOval}.
  Finally we remark that the algorithm may not detect curve jumping errors
  caused by exchanging branches during tracing.
\end{remark}

\begin{algorithm}
\LinesNumbered
\caption{${\sf ApproxPlot}$}
\label{Algo:Plot}
\KwIn{
  A finite set of polynomials $F=\{f_1,\ldots,f_{n-1}\}\in\R[x_1,\ldots,x_n]$.
  A bounding box $B\subset \R^n$.
  A precision $\epsilon>0$.
}
\KwOut{
  An $\epsilon$-approximation of $F^{-1}(0)$ in $B$.
}
\Begin{
    \uIf{$F$'s coefficients are known exactly}{
      let $S_0={\sf SingularPoints}(F, B)$\;
    }
    \Else{
      let $S_0={\sf PseudoSingularPoints}(F, B, \epsilon)$\;
    }
    let ${\cal C}, \delta := {\sf NatualClusters}(S_0,\epsilon/2)$\;
    $cwp := \emptyset$; $bwp := \emptyset$\;
    \For{each natural cluster $C\in{\cal C}$}{
      let $p$ be the center point of $C$\;
      \For{each $q\in {\sf FencingPoints}(p,\delta)$}{
        let $s$ be the smallest singular value of ${\Jac}_F(q)$\;
        let $v := q-p$; let $c := 0$;
        add $(q, v, s, c)$ to $cwp$\;
      }
    }
    \For{each point $q$ of $F^{-1}(0)\cap \partial B$}{
      let $s$ be the smallest singular value of ${\Jac}_F(q)$\;
      let $v$ be the tangent vector of $F^{-1}(0)$ towards the interior of $B$\;
      let $c:=0$;
      add $(q, v, s, c)$ to $bwp$\;
    }
    let $\Delta$ be the union of balls associated with the natural clusters\;
    rescale the coefficients of $F$ if necessary\;
    set $rwp:={\sf WitnessPoints}(F, B)\setminus\Delta$\;
    \tcc{Note that below the function {\sf PlotMain} is called multiple times with different arguments and flags.}
    $S_1, front := {\sf PlotMain}(F, B, cwp, bwp, rwp, \delta, try)$\;
    $S_2 := {\sf PlotMain}(F, B, front, bwp, rwp, \delta, resume)$\;    
    $S_3 := {\sf PlotMain}(F, B, bwp, cwp, rwp, \delta, boundary)$\;
    $S_4 := {\sf PlotOval}(F, B, rwp, cwp\cup bwp, \delta)$\;
    \uIf{$F$'s coefficients are known exactly}{
      return $\{\cup_{i=0}^4 S_i\}$\;
    }
    \Else{
      return  $\{\cup_{i=1}^4 S_i\}$\;
    }
}
\end{algorithm}

\begin{algorithm}
    \LinesNumbered
  \caption{${\sf PlotMain}(F, B, cwp, bwp, rwp, \delta, tag)$}
  \label{Algo:main}
  \Begin{
     $S := \emptyset$; $front := \emptyset$\;
      \For{$j$ to $|cwp|$}{
          $P := \emptyset$; 
          $(q, v, s, c) := cwp[j]$\;
          {\bf if} {$cwp[j].c>0$} {\bf then} {next;} {\bf else} $cwp[i].c := 1$\;
          $mb := false$; $mc := false$;  $mf := false$\;
      
          \While{$q\in B$}{   
            
            $s' := s$;
            $q' := q$;
            $v' := v$;       
            $P := P\cup \{q\}$\;    
            choose step size $h\leq \delta/2$ according to $\delta$ and $s$\;
            $q := q+hv$;\tcp{Predictor step}
      
            with $q$ as initial point, apply Newton iterations to update $q$\;
            let $s$ be the smallest singular value of $J_F(q)$;
            let $v := {\frak t}(J_F(q))$\;       
            \lIf{$v \bullet v'<0$}{ $v := -v$ }    
            remove any element of $rwp$ on $\overline{q'q}$\;
            \For{$i$ to $|bwp|$}{
                 \If{$(bwp[i].v)\bullet v<0 \wedge bwp[i].q \in\overline{q'q}$}{
                   \uIf{$bwp[i].c>0$}{
                       report curve jump error\;
                   }
                   \Else{
                       $P := P\cup \{bwp[i].q\}$\;
                   }
                   $mb := true$;
                   $bwp[i].c := bwp[i].c + 1$;
                   break\;
                 }
            }
           \lIf{$mb$}{break}
            
            \For{$i$ to $|cwp|$}{
                 \If{$i\neq j$ and $(cwp[i].v)\bullet v<0$ and $cwp[i].q \in\overline{q'q}$}{
                   \uIf{$cwp[i].c>0$}{
                       report curve jump error\;
                   }
                   \ElseIf{tag is 'resume' or 'try'}{
                       $P := P\cup \{cwp[i].q\}$\;
                   }
                   $mc := true$;
                   $cwp[i].c := cwp[i].c + 1$;
                   break\;
                 }
            }
            \lIf{$mc$}{break}
          
            \If{tag='try'}{
                \For{$i$ to $|front|$}{
                     \If{$front[i].v\bullet v<0$ and $front[i].q \in\overline{q'q}$}{
                       $P := P\cup \{front[i].q\}$;
                       $mf := true$;
                       remove $front[i]$ from $front$\;
                       break\;
                     }
                }  
                \lIf{$mf$}{break}

                \If{$s < s'$}{
                    add $(q, v, s, 0)$ to $front$;break\;
                }
	    }
          }
          $S := S\cup \{P\}$\;
      }
      {\bf if} {tag='try'} {\bf then} {return $S$, $front$} {\bf else} {return $S$};
  }
\end{algorithm}

  \begin{algorithm}
  \LinesNumbered
  \caption{${\sf PlotOval}(F, B, rwp, wp, \delta)$}
  \label{Algo:oval}
  \Begin{
      $S := \emptyset$\;

      \While{$rwp\neq\emptyset$}{
          $P := \emptyset$\;
          choose $p\in rwp$ and set $rwp := rwp\setminus \{p\}$;
          $k := 0$;       
          $q := p$\;
          let $s$ be the smallest singular value of $J_F(q)$\;
          let $v := {\frak t}(J_F(q))$\; 
          
          $mt := false$\;
      
          \While{$q\in B$}{   

            $k := k+1$;            
            $q' := q$;
            $v' := v$;       
            $P := P\cup \{q\}$\;            
	    
            choose step size $h\leq \delta/2$ according to $\delta$ and $s$\;

            $q := q+hv$\;
      
            with $q$ as initial point, apply Newton iterations to update $q$\;

            let $s$ be the smallest singular value of $J_F(q)$\;
            let $v := {\frak t}(J_F(q))$\;   
            
            \lIf{$v\bullet v'<0$}{ $v := -v$ }
            \If{$k>2$ and $p\in \overline{q'q}$}{
               break\;
            }
            remove any element of $rwp$ other than $p$ on $\overline{q'q}$\;
            \For{$i$ to $|wp|$}{
                 \If{$(wp[i].v)\bullet v<0 \wedge wp[i].q \in\overline{q'q}$}{
                   \If{$wp[i].c>0$}{
                       report curve jump error\;
                   }
                   $mt := true$;
                   $wp[i].c := wp[i].c + 1$;
                   break\;
                 }
            }
           \lIf{$mt$}{break}

          }
          $S := S\cup \{P\}$\;
      }
      return $S$\;
  }
\end{algorithm}

\section{Experimentation}
\label{sec:exp}
In this section, we provide some nontrivial examples
to illustrate the effectiveness of our method.
Example~\ref{ex:discrim} is the discriminant of a random trivariate polynomial.
Example~\ref{ex:res} is the resultant of two random trivariate polynomials.
To make a fair comparison with the {\sf Plots:-implicitplot} command of {\sf Maple} $18$,
all polynomials are plotted using their irreducible factors.

We have implemented our algorithm in Maple.
In the algorithms of last section, there are several places where
ones needs to solve zero-dimensional polynomial systems,
namely computing singular points, computing fencing points around singular points,
computing boundary points and computing witness points.
For the first three, we find that it is more robust to
call a symbolic solver and use {\sf RootFinding:-Isolate}
of Maple. For the last one, we find it is more efficient
to use homotopy based methods and we implemented a Maple
interface to {\sf hom4ps2}~\cite{Lee2008}.
The experimental
results in this section were obtained on an Fedora laptop (Intel i7-7500U CPU @ 2.70GHz, 16.0 GB total memory).
A preliminary implementation of {\sf ApproxPlot} is available at
\url{http://www.arcnl.org/cchen/software}.

\subsection{Visualizing planar curve}

\begin{example}
  \label{ex:discrim}
  Let $f$ be the same polynomial as in Example~\ref{ex:discrim3}.
  Visualizations of $it$ by {\sf Plot:-implicitplot} in Maple and {\sf ApproxPlot} is depicted in Fig.~\ref{fig:discrim}.
  The option for {\sf Plot:-implicitplot} is \verb+grid=[300,300], gridrefine=6, crossingrefine=4+ and
  the time spent is $34.3$ seconds.
  The option for {\sf ApproxPlot} is $\epsilon=0.5$ and the time spent is $42.7$ seconds.
  For this example, the initial radius $\epsilon/2$ suffices to dividing the $4$ singular points into $3$
  natural clusters.
  The polynomial $f$ has branches very close to each other and the algorithm detects curve jumping.
\end{example}

\begin{figure}
\begin{minipage}[t]{0.45\linewidth}
\centering
\includegraphics[width=\textwidth]{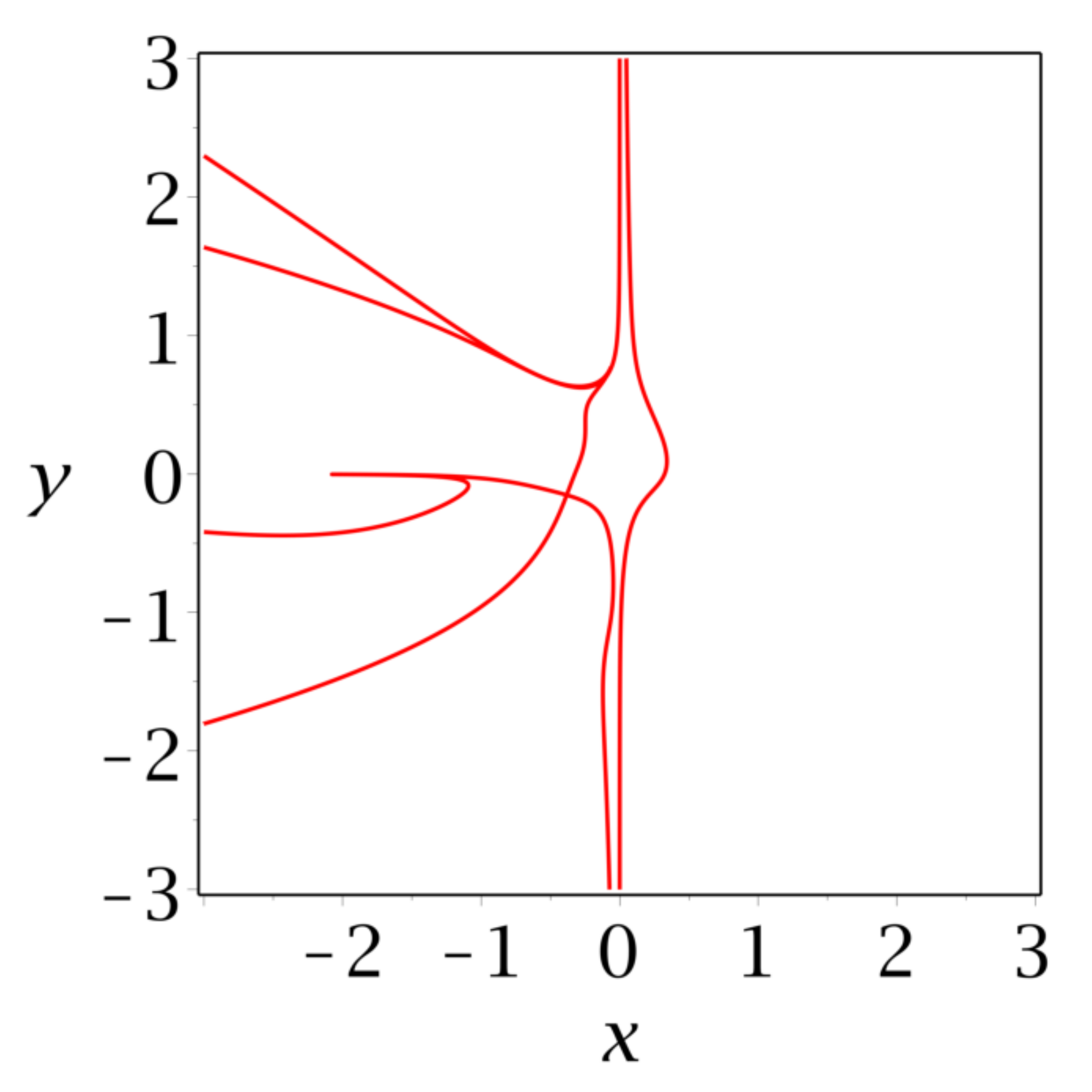}
{By Maple {\sf Plot:-implicitplot}.}
\end{minipage}
\begin{minipage}[t]{0.45\linewidth}
\centering
\includegraphics[width=\textwidth]{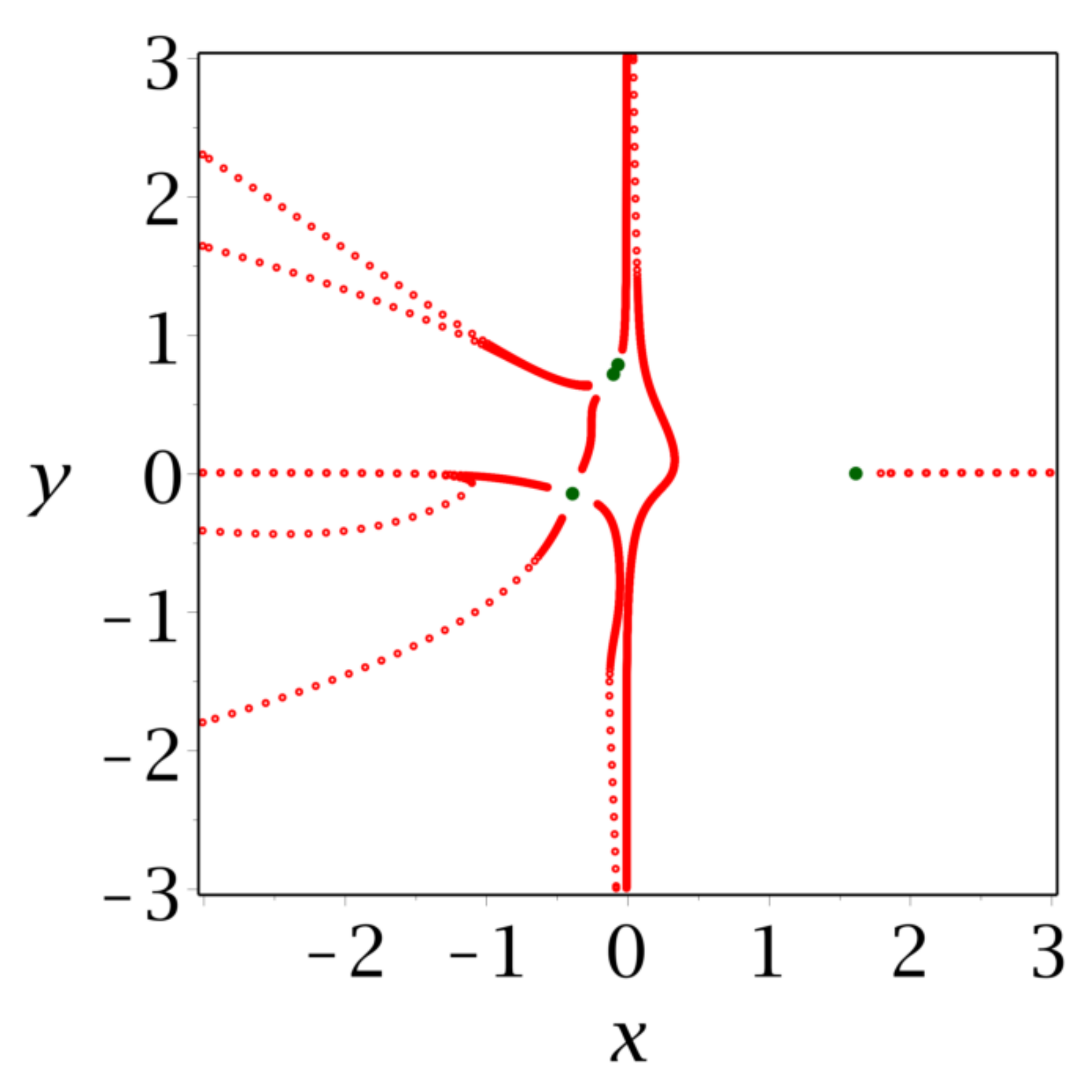}
{By {\sf ApproxPlot}.}
\end{minipage}
  \caption{Visualization of Example~\ref{ex:discrim}.}
  \label{fig:discrim}
\end{figure}

\begin{example}
  \label{ex:res}
  Let $f_1 := 72\,{y}^{2}{z}^{5}+26\,{x}^{2}y{z}^{3}-84\,{x}^{2}{y}^{2}-73\,x{z}^{2}+6,
  f_2 := -24\,{x}^{4}{z}^{2}-35\,y{z}^{3}+43\,y{z}^{2}-66\,{z}^{3}+3$.
  Let $f$ be the resultant of $f_1$ and $f_2$ w.r.t. $z$.
  A visualization of it is depicted in Fig.~\ref{fig:res}.
  The polynomial $f$ has branches very close to each other and the algorithm detects curve jumping.
  The option for {\sf Plot:-implicitplot} is \verb+grid=[300,300],gridrefine=6,crossingrefine=6+ and
  the time spent is $36.4$ seconds.
  The option for {\sf ApproxPlot} is $\epsilon=0.5$ and the time spent is $32.5$ seconds.
  For this example, the initial radius for finding clusters is $0.25$,
  which finally gets updated to $0.0625$ to find $7$ natural clusters for $10$ singular points,
  see Fig.~\ref{fig:cluster}.
\end{example}

\begin{figure}
\begin{minipage}[t]{0.45\linewidth}
\centering
\includegraphics[width=\textwidth]{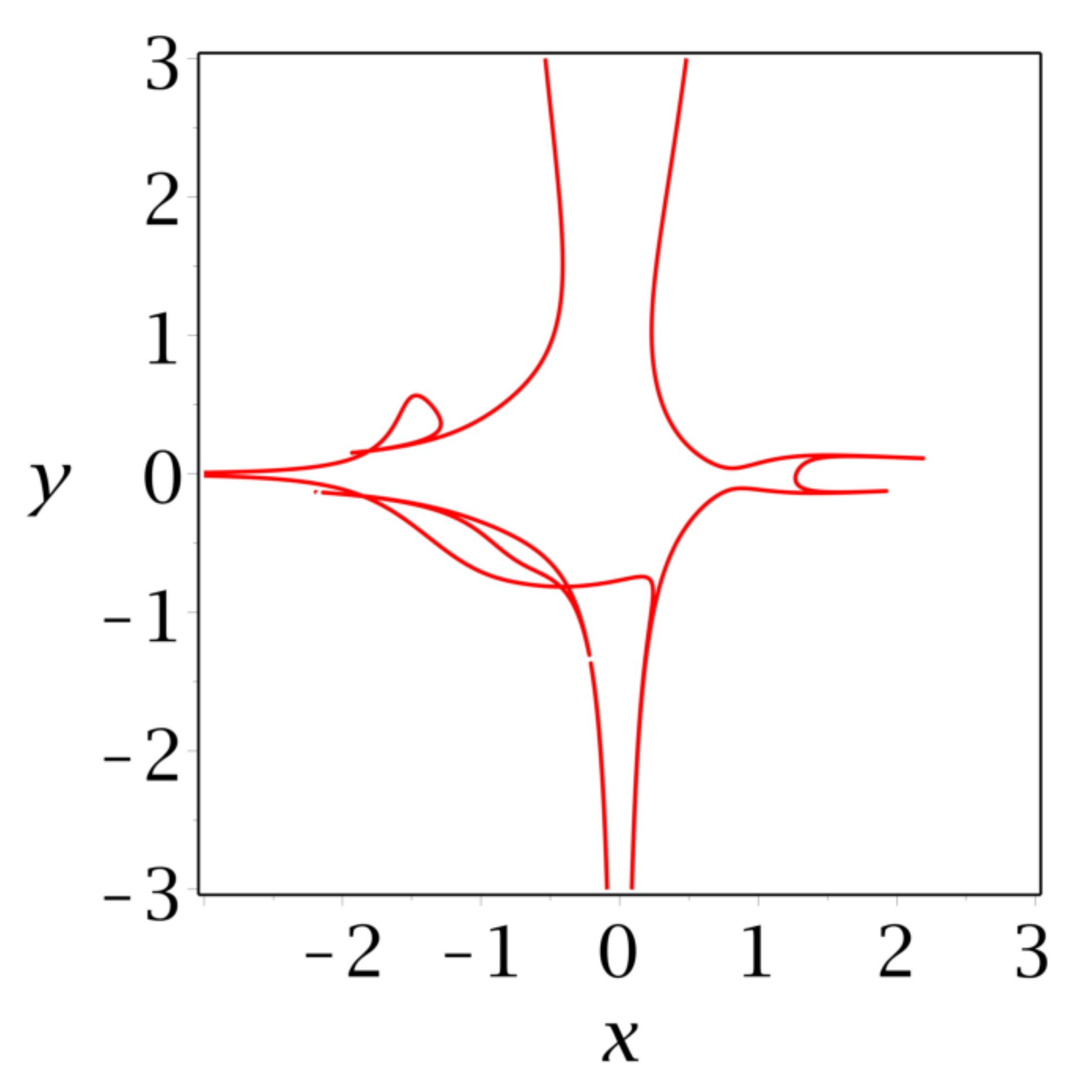}
{By Maple.}
\label{fig:plot11}
\end{minipage}
\begin{minipage}[t]{0.45\linewidth}
\centering
\includegraphics[width=\textwidth]{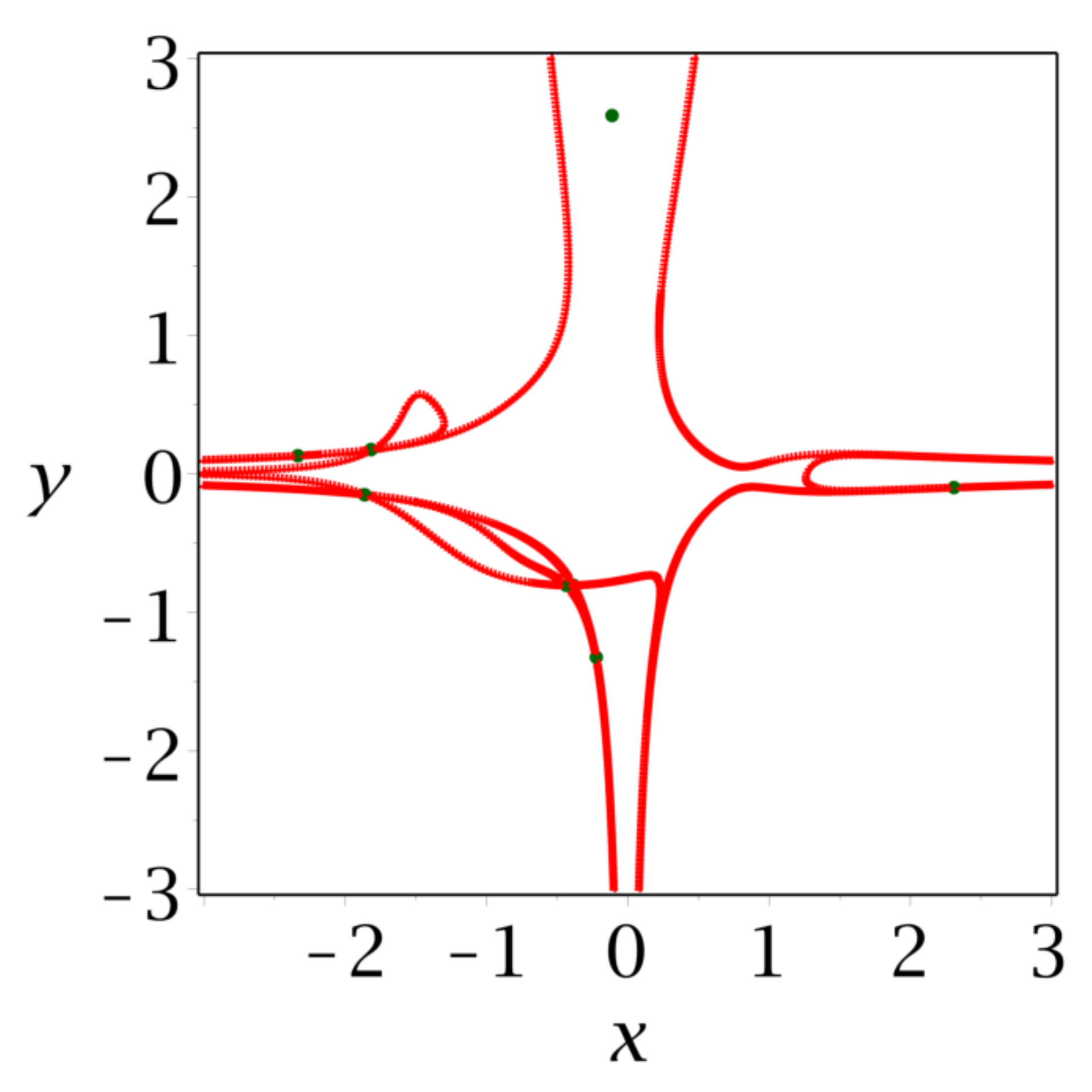}
{By {\sf ApproxPlot}.}
\label{fig:plot12}
\end{minipage}
  \caption{Visualization of Example~\ref{ex:res}.}
  \label{fig:res}
\end{figure}

\begin{figure}
\begin{minipage}[t]{0.45\linewidth}
\centering
\includegraphics[width=\textwidth]{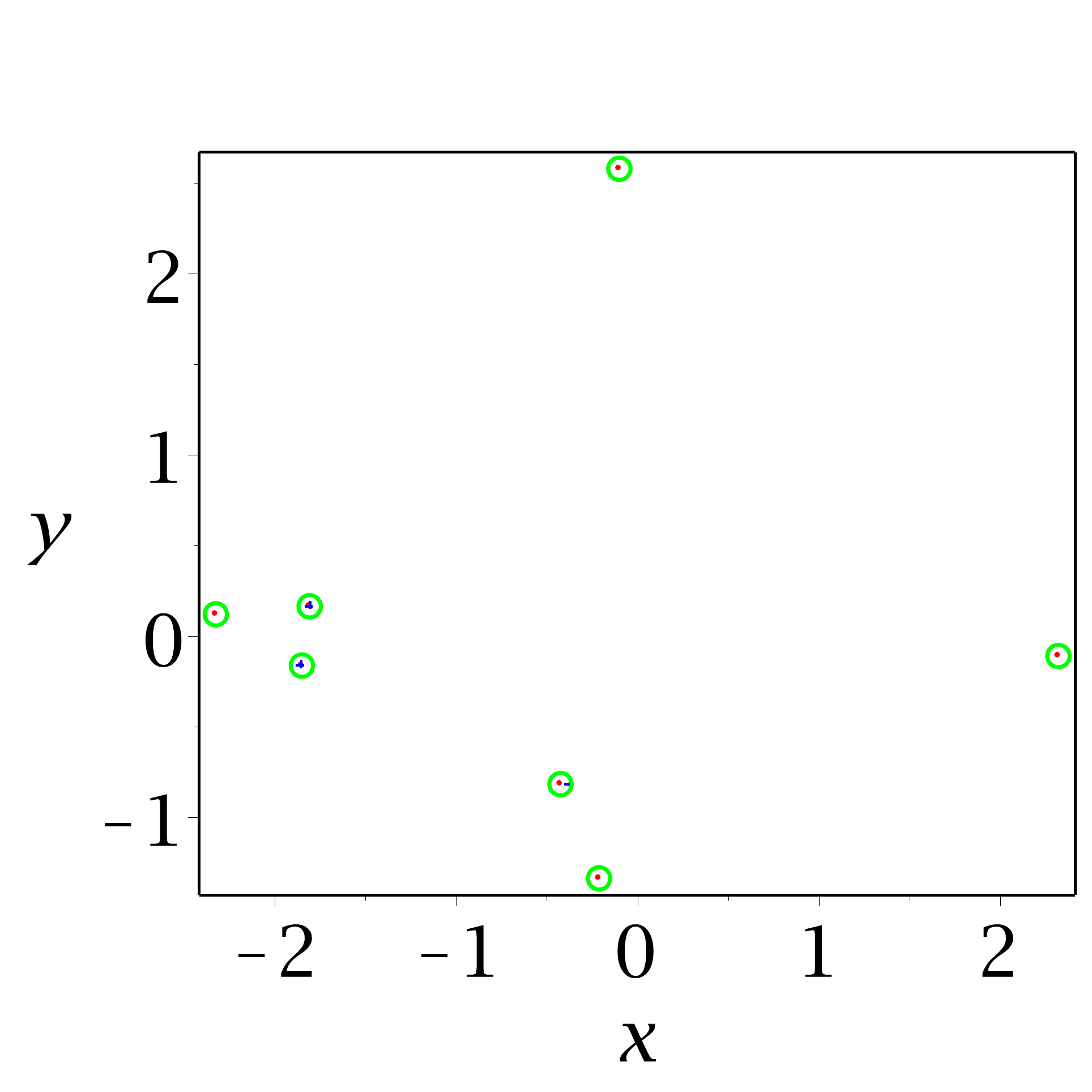}
{The seven natural clusters.}
\end{minipage}
\begin{minipage}[t]{0.45\linewidth}
\centering
\includegraphics[width=\textwidth]{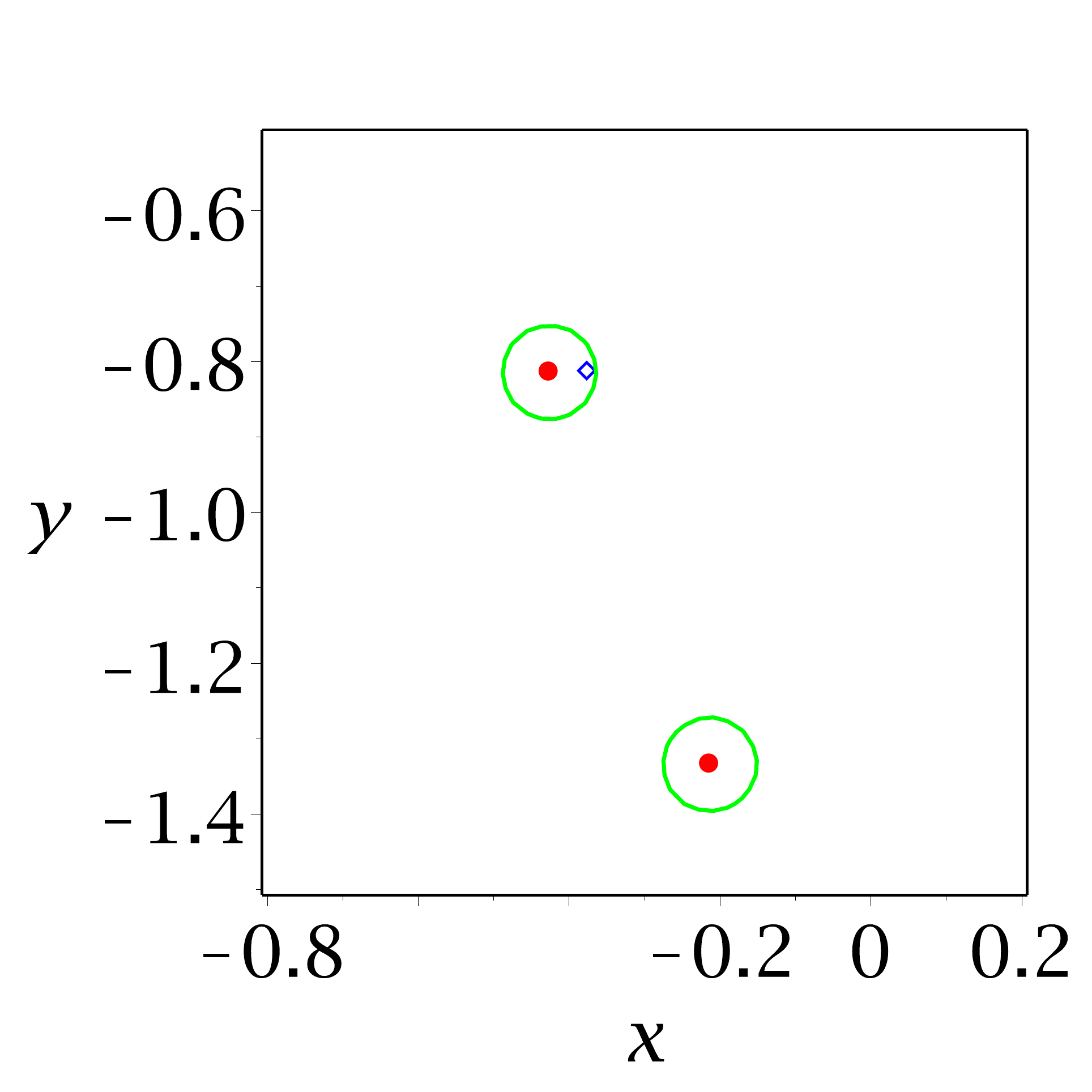}
{A close-up view of two natural clusters.}
\end{minipage}
  \caption{Natural clusters for Example~\ref{ex:res}.}
  \label{fig:cluster}
\end{figure}

\subsection{Visualizing space curve}

\begin{example}
  \label{ex:random3d}
  Let $f := \{3\,{x}^{7}{y}^{4}z+{y}^{12}+{z}^{12}-{y}^{5}{z}^{6}-5\,{x}^{7}{y}^{2}z+5\,x{y}^{3}{z}^{6}-8\,{x}^{3}yz-8\,{y}^{3}-y{z}^{2}+10\}$
  be a randomly generated trivariate polynomial.
  Let $\displaystyle F := \bigg\{\frac{\partial f}{\partial z}, \frac{\partial f}{\partial y}\bigg\}$.
  A visualization of it is depicted in Fig.~\ref{fig:random3d}.
  The option for {\sf plots:-intersectplot} is $grid=[60,60,60]$ and the time spent is $33.5$ seconds.
  The option for {\sf ApproxPlot} is $\epsilon=0.2$ and the time spent is $22.2$ seconds.
  No curve jumping is reported by {\sf ApproxPlot}.
  As we can see, {\sf Maple } has issues when visualizing the parts on the plane $y=-3$ and $y=3$.
\end{example}

\begin{figure}
\begin{minipage}[t]{0.46\linewidth}
\centering
\includegraphics[width=\textwidth]{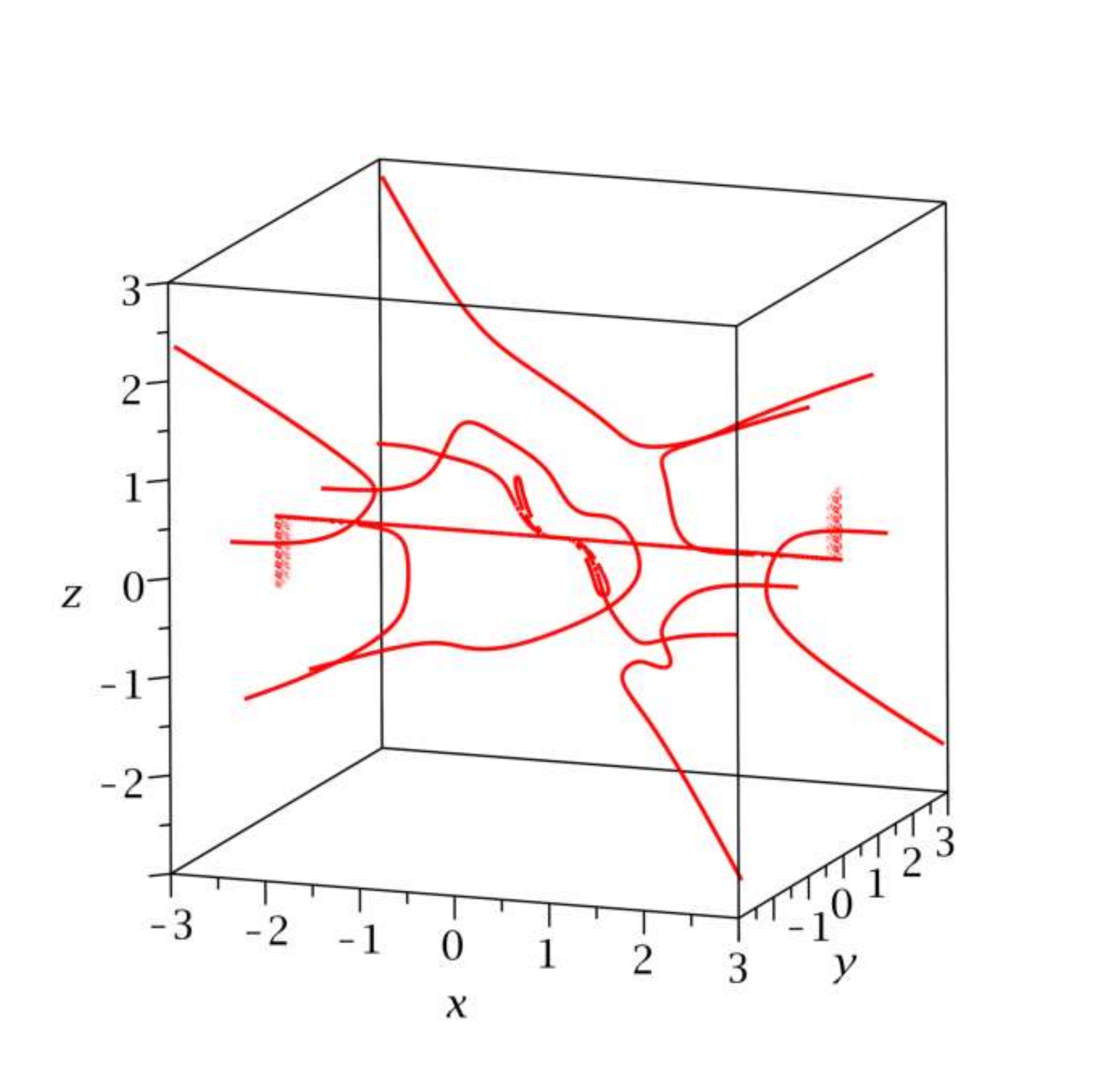}
{By Maple {\sf plots:-intersectplot}.}
\end{minipage}
\begin{minipage}[t]{0.51\linewidth}
\centering
\includegraphics[width=\textwidth]{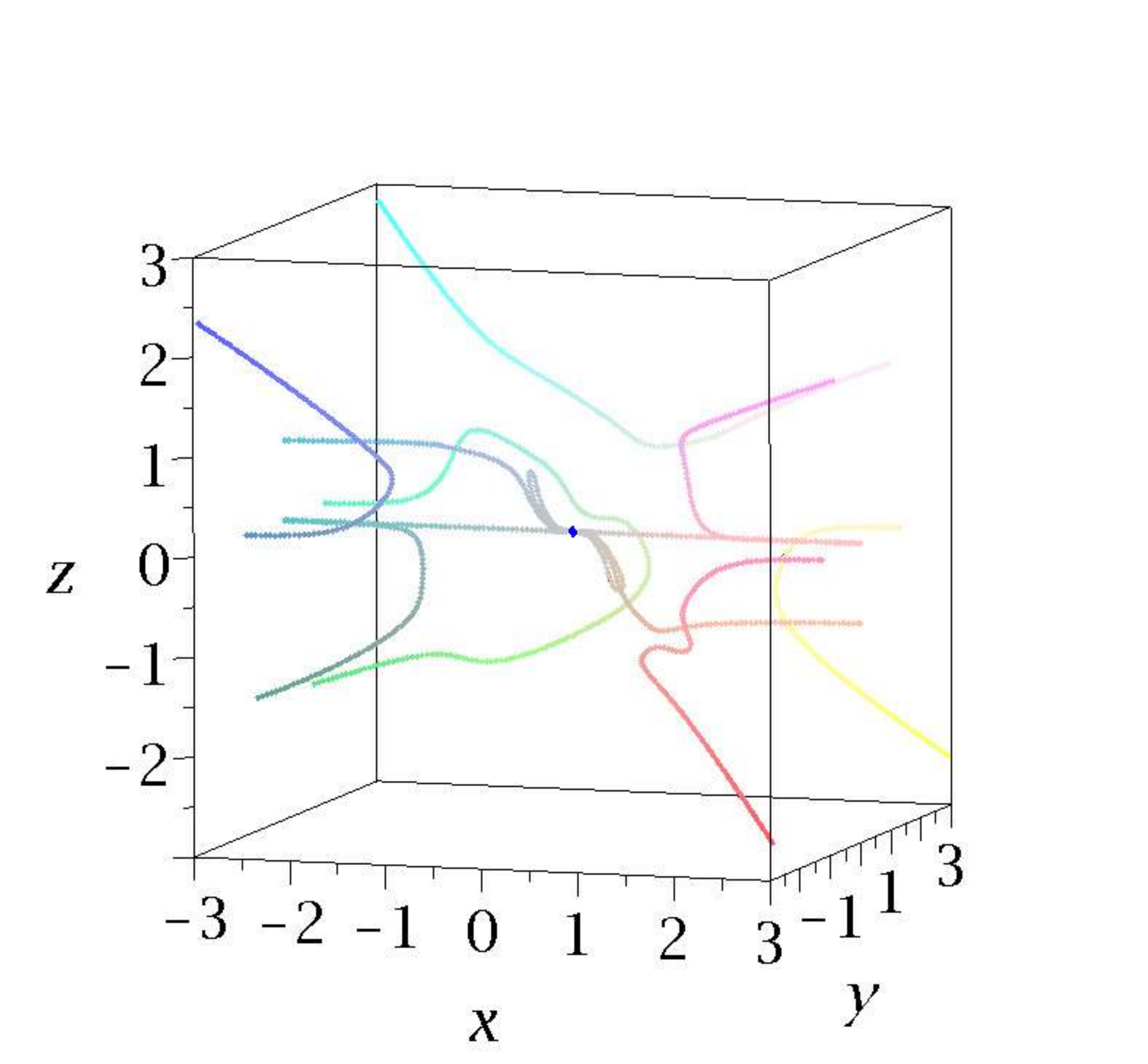}
{By {\sf ApproxPlot}.}
\end{minipage}
   \caption{Visualization of Example~\ref{ex:random3d}.}
  \label{fig:random3d}
\end{figure}

\begin{example}
  \label{ex:random3d-2}
  Let $f := -51\,{x}^{12}yz+31\,{x}^{4}y{z}^{8}+68\,{x}^{11}z-10\,{x}^{5}{y}^{4}{z}^{3}+{z}^{12}+91\,{x}^{5}{y}^{2}{z}^{3}-81\,{y}^{5}+10$
  and $\displaystyle F := \bigg\{f, \frac{\partial f}{\partial z}\bigg\}$.
   A visualization of it is depicted in Fig.~\ref{fig:random3d-2}.
  The option for {\sf plots:-intersectplot} is $grid=[60,60,60]$ and the time spent is $40.5$ seconds.
  The option for {\sf ApproxPlot} is $\epsilon=0.2$ and the time spent is $27.3$ seconds.
  No curve jumping is reported by {\sf ApproxPlot}.
  As we can see, {\sf Maple } clearly has visualization issues.
\end{example}

\begin{figure}
\begin{minipage}[t]{0.46\linewidth}
\centering
\includegraphics[width=\textwidth]{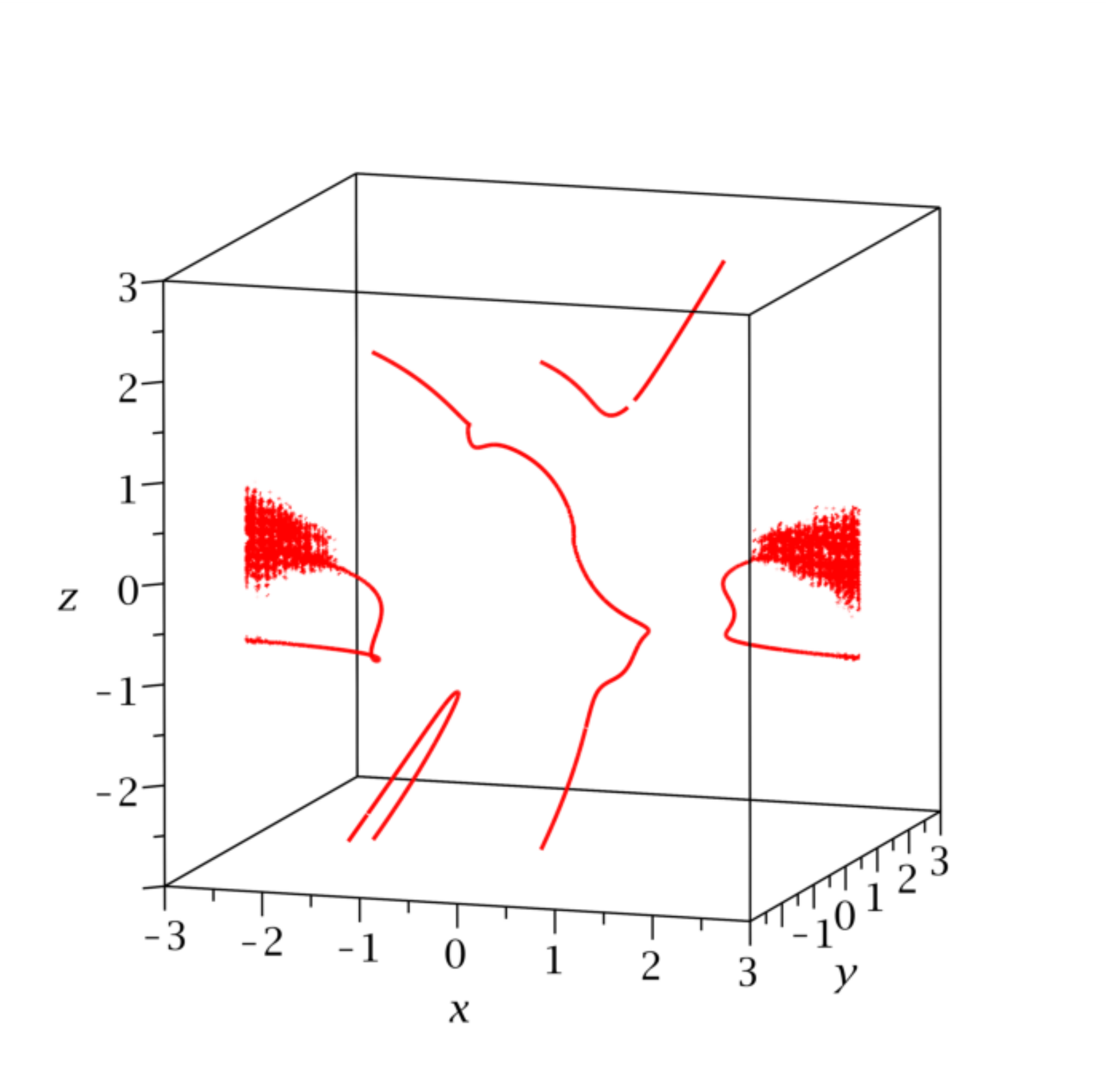}
{By Maple {\sf plots:-intersectplot}.}
\end{minipage}
\begin{minipage}[t]{0.51\linewidth}
\centering
\includegraphics[width=\textwidth]{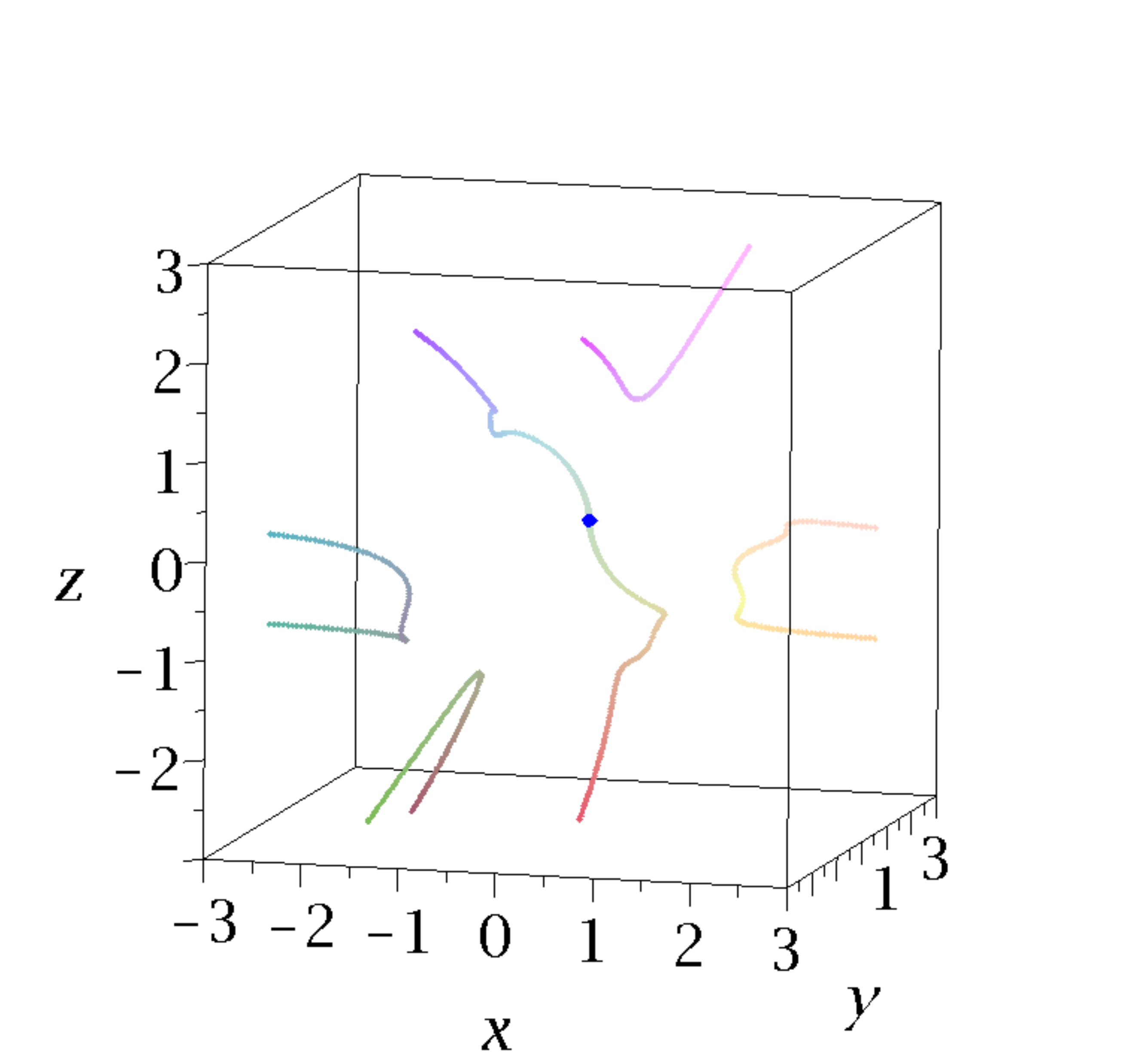}
{By {\sf ApproxPlot}.}
\end{minipage}
   \caption{Visualization of Example~\ref{ex:random3d-2}.}
  \label{fig:random3d-2}
\end{figure}

\begin{example}
  \label{ex:BarthDecic}
  We consider a famous surface called Barth Decic~\cite{Duzhin} defined
  by
  $$
  \begin{array}{rl}
    f := &8\, \left( -{r}^{4}{y}^{2}+{x}^{2} \right)  \left( -{r}^{4}{z}^{2}+{y}^{2} \right)\left( -{r}^{4}{x}^{2}+{z}^{2} \right)
    \left( {x}^{4}-2\,{x}^{2}{y}^{2}-2\,{x}^{2}{z}^{2}+{y}^{4}-2\,{y}^{2}{z}^{2}+{z}^{4}\right)\\
    &+ \left( 3+5\,r \right)  \left( -{w}^{2}+{x}^{2}+{y}^{2}+{z}^{2} \right) ^{2}
    \left( {x}^{2}+{y}^{2}+{z}^{2}- \left( 2-r \right) {w}^{2} \right) ^{2}{w}^{2}
  \end{array}
  $$
  and $\displaystyle F := \bigg\{f, \frac{\partial f}{\partial z}\bigg\}$.
  
  We consider the case that $w=1$ and $r=2$.
  The polynomial $\displaystyle \frac{\partial f}{\partial z}$ has two irreducible factors, one of them is $z$.
  Let $F_1=[f, z]$ and $\displaystyle F_2=\bigg[f, 1/z\cdot\frac{\partial f}{\partial z}\bigg]$.
  We plot $F_1$ and $F_2$ separately and display them together.
  The option for {\sf plots:-intersectplot} is $grid=[80,80,80]$
  and the total time spent is $108.7+154.8=263.5$ seconds.
  
   The option for {\sf ApproxPlot} is $\epsilon=0.1$ and the time spent is $1704.3$ seconds.
  No curve jumping is reported by {\sf ApproxPlot}.
   A visualization of it is depicted in Fig.~\ref{fig:BarthDecic}.
  As we can see, {\sf Maple } clearly has visualization issues.
\end{example}

\begin{figure}
\begin{minipage}[t]{0.48\linewidth}
\centering
\includegraphics[width=\textwidth]{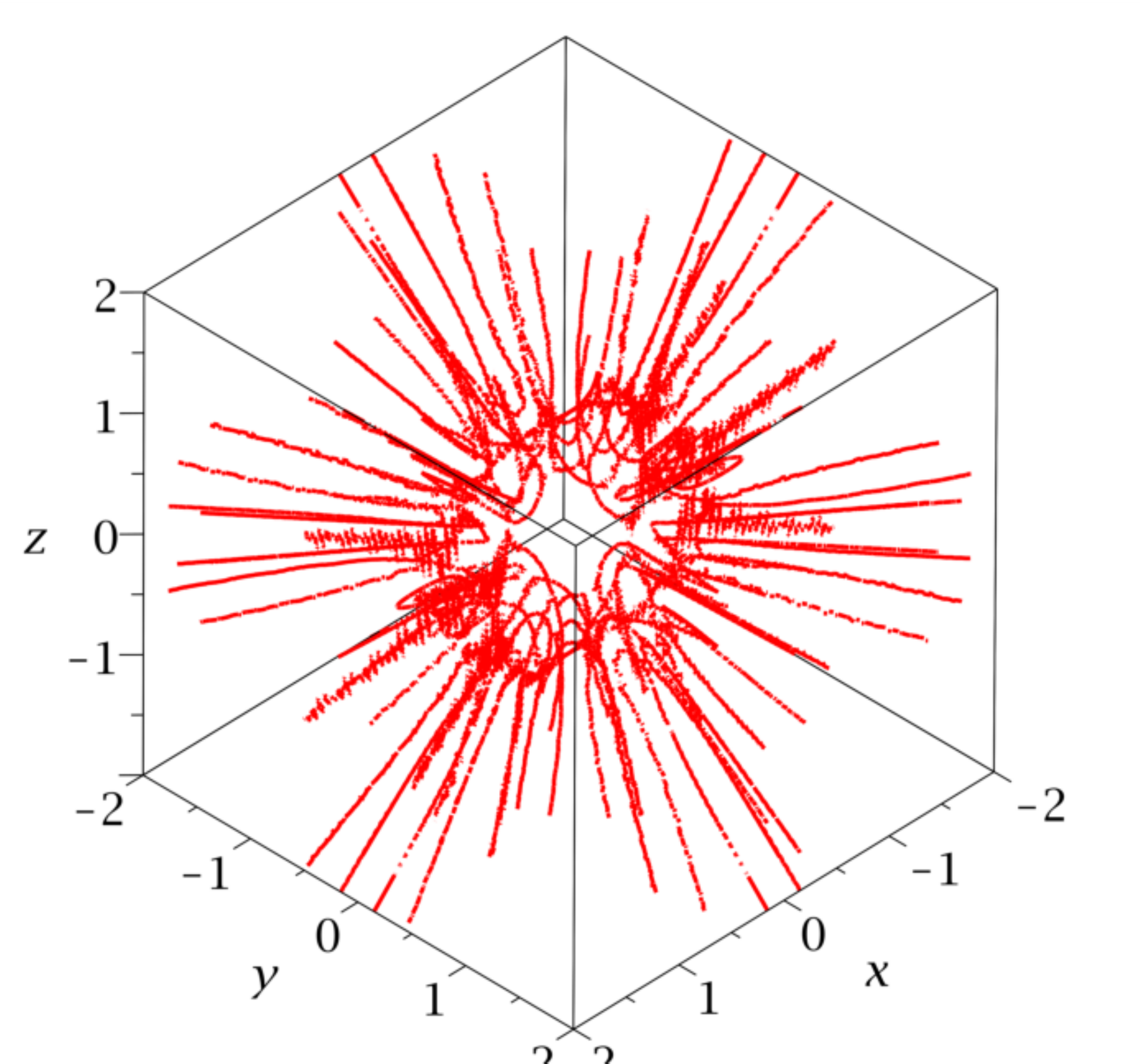}
{By Maple {\sf plots:-intersectplot}.}
\end{minipage}
\begin{minipage}[t]{0.5\linewidth}
\centering
\includegraphics[width=\textwidth]{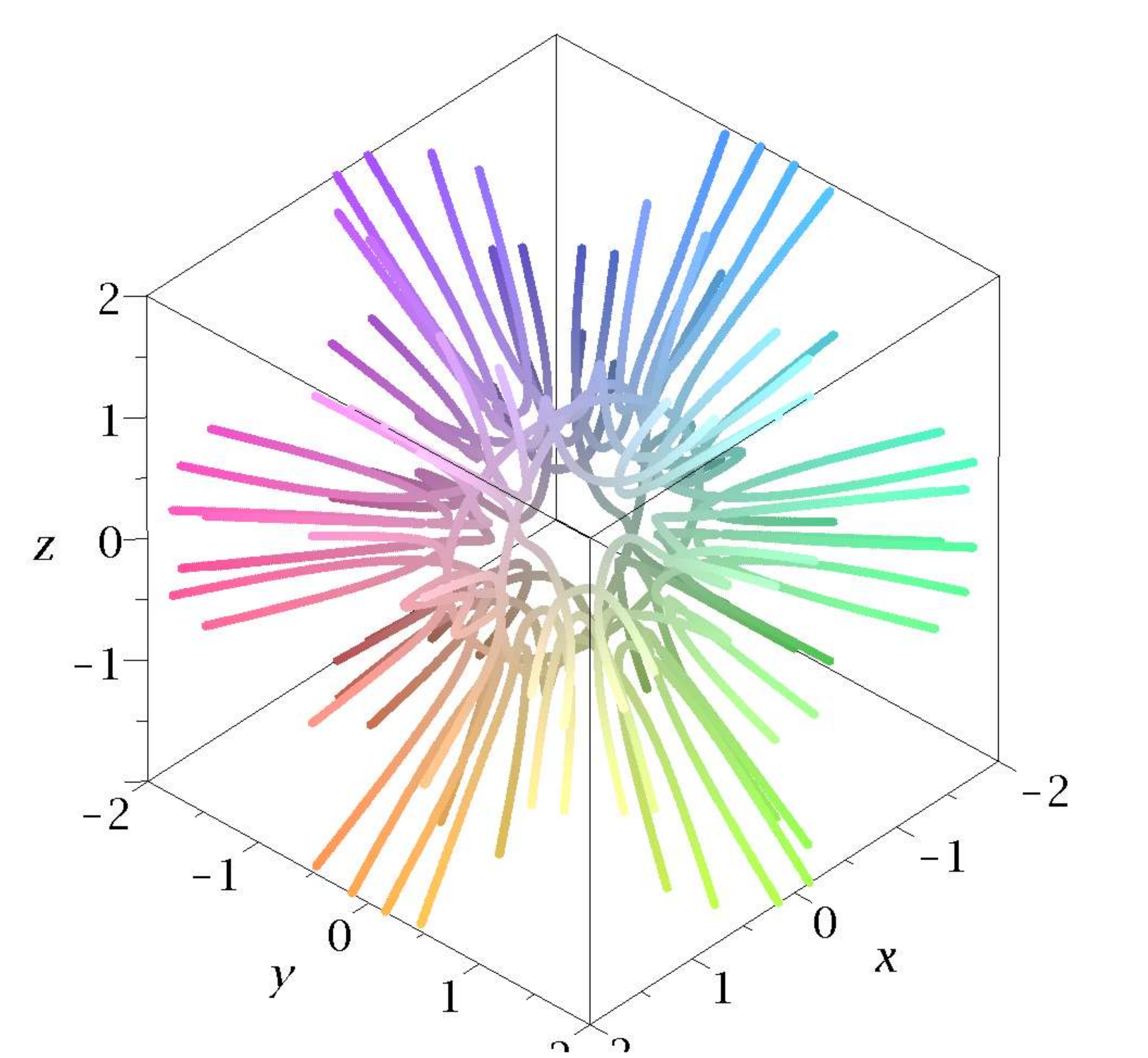}
{By {\sf ApproxPlot}.}
\end{minipage}
   \caption{Visualization of Example~\ref{ex:BarthDecic}.}
  \label{fig:BarthDecic}
\end{figure}

\begin{example}
  \label{ex:BarthSextic}
  We consider a famous surface called Barth Sextic~\cite{Duzhin} defined
  by $f := 4\, \left( {r}^{2}{x}^{2}-{y}^{2} \right)  \left( {r}^{2}{y}^{2}-{z}^{2} \right)  \left( {r}^{2}{z}^{2}-{x}^{2} \right)
  - \left( 1+2\,r\right)  \left( -{w}^{2}+{x}^{2}+{y}^{2}+{z}^{2} \right) ^{2}{w}^{2}$
  and $\displaystyle F := \bigg\{f, \frac{\partial f}{\partial z}\bigg\}$.  
  Consider the case that $w=1$ and $\displaystyle r=\frac{\sqrt{5}+1}{2}$.
  The option for {\sf plots:-intersectplot} is $grid=[100,100,100]$, $maxlev=8$ and the time spent is $107.7$ seconds.
  The option for {\sf ApproxPlot} is $\epsilon=0.1$.
  There are two ways to plot the curve by {\sf ApproxPlot}.
  One is to use an approximate value of $r$ and plot the curve in $\R^3$.
  In such case, an exact computation shows that there are no singular points.
  As a result, there is curve jumping and some part of the curve is missing after tracing.
  Instead, we use the technique of pseudo singular points, and no curve jumping is reported.
   The time spent is $76.7$ seconds.
  Another method is to encode the algebraic number $r$ by its defining polynomial $r^2-r-1$.
  We then trace the curve defined by $\displaystyle \bigg\{f, \frac{\partial f}{\partial z}\bigg\}, r^2-r-1$ in $\R^4$
  and take its projection in $\R^3$.
   The time spent in this way is $118.3$ seconds.
  No curve jumping is reported by {\sf ApproxPlot}.
  A visualization of the space curve is depicted in Fig.~\ref{fig:BarthSextic}.
\end{example}

\begin{example}
  \label{ex:EndrassOctic}
  Consider $Endra\ss$ Octic surface~\cite{Duzhin} defined by
  $$
  \begin{array}{rl}
    f  := &64\, \left( -{w}^{2}+{x}^{2} \right)  \left( -{w}^{2}+{y}^{2} \right) \left(  \left( x+y \right) ^{2}-2\,{w}^{2} \right)\left(  \left( x-y \right) ^{2}-2\,{w}^{2} \right)\\
    &- \bigg( -4\, \left( 1+\sqrt {2}\right)  \left( {x}^{2}+{y}^{2} \right) ^{2}+ \left( 8\, \left( 2+\sqrt {2} \right) {z}^{2} +2\, \left( 2+7\,\sqrt {2} \right) {w}^{2}\right)\left( {x}^{2}+{y}^{2} \right)\\
    &-16\,{z}^{4}+8\, \left( 1-2\,\sqrt {2} \right) {z}^{2}{w}^{2}- \left( 1+12\,\sqrt {2} \right) {w}^{4} \bigg) ^{2}
  \end{array}
  $$
    and $\displaystyle F := \bigg\{f, \frac{\partial f}{\partial z}\bigg\}$.  
    We consider the case that $w=1$.
    The algebraic number $\sqrt{2}$ is replaced by its approximation.
  The option for {\sf plots:-intersectplot} is $grid=[100,100,100]$ and the time spent is $219.3$ seconds.
  For {\sf ApproxPlot}, 
  we use the technique of pseudo singular points.
   Choosing $\epsilon=0.1$, the time spent is $643.1$ seconds with curve jumping reported.
   Setting $\epsilon=0.05$, the time spent is $1921.6$ seconds and no curve jumping is reported.
   Visualization of the curve using the two values of $\epsilon$ is similar and is depicted in Fig.~\ref{fig:EndrassOctic}.
\end{example}

\begin{figure}
\begin{minipage}[t]{0.48\linewidth}
\centering
\includegraphics[width=\textwidth]{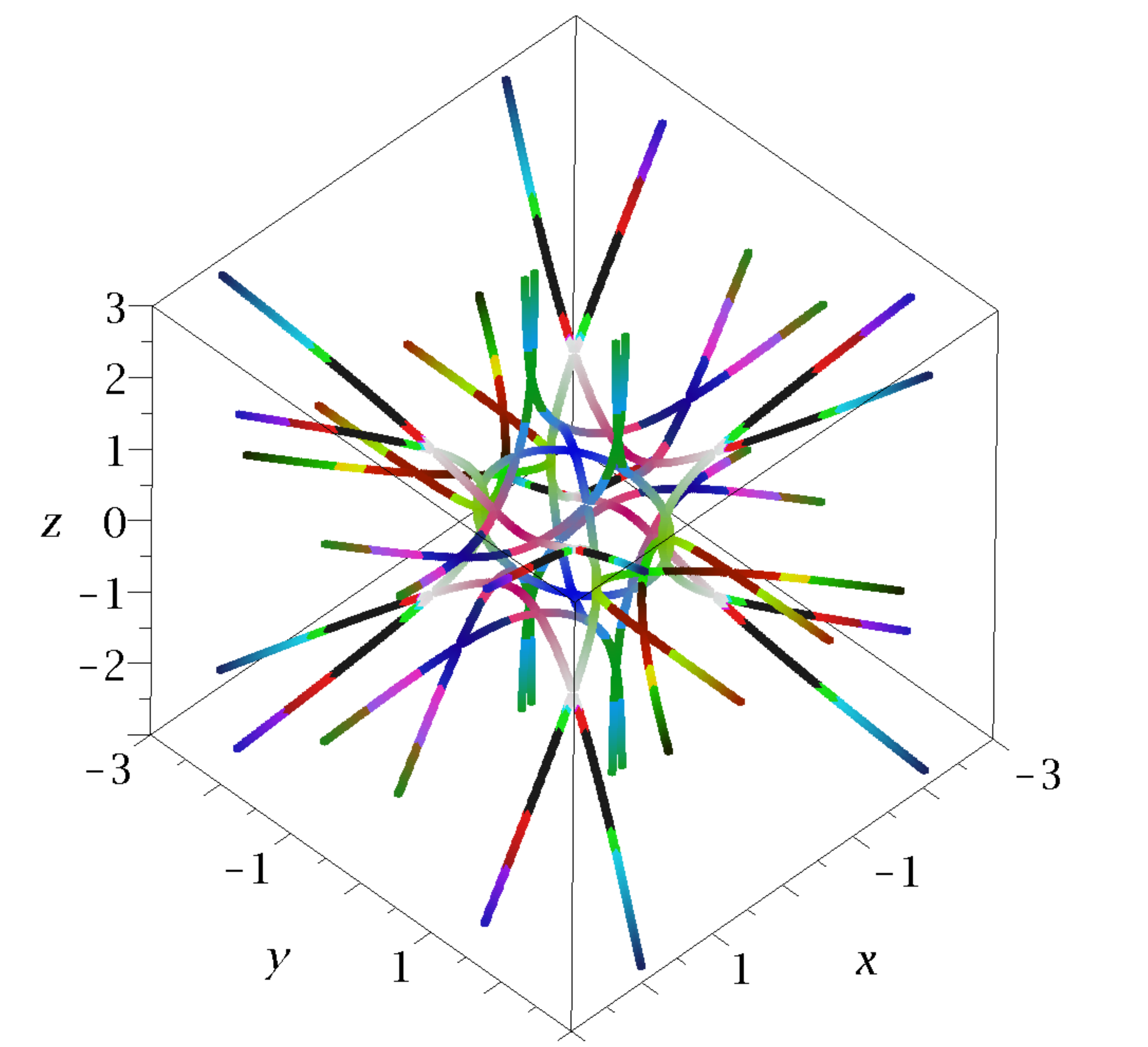}
{By {\sf ApproxPlot} ({\sf plots:-intersectplot} can achieve similar effect).}
\end{minipage}
\begin{minipage}[t]{0.43\linewidth}
\centering
\includegraphics[width=\textwidth]{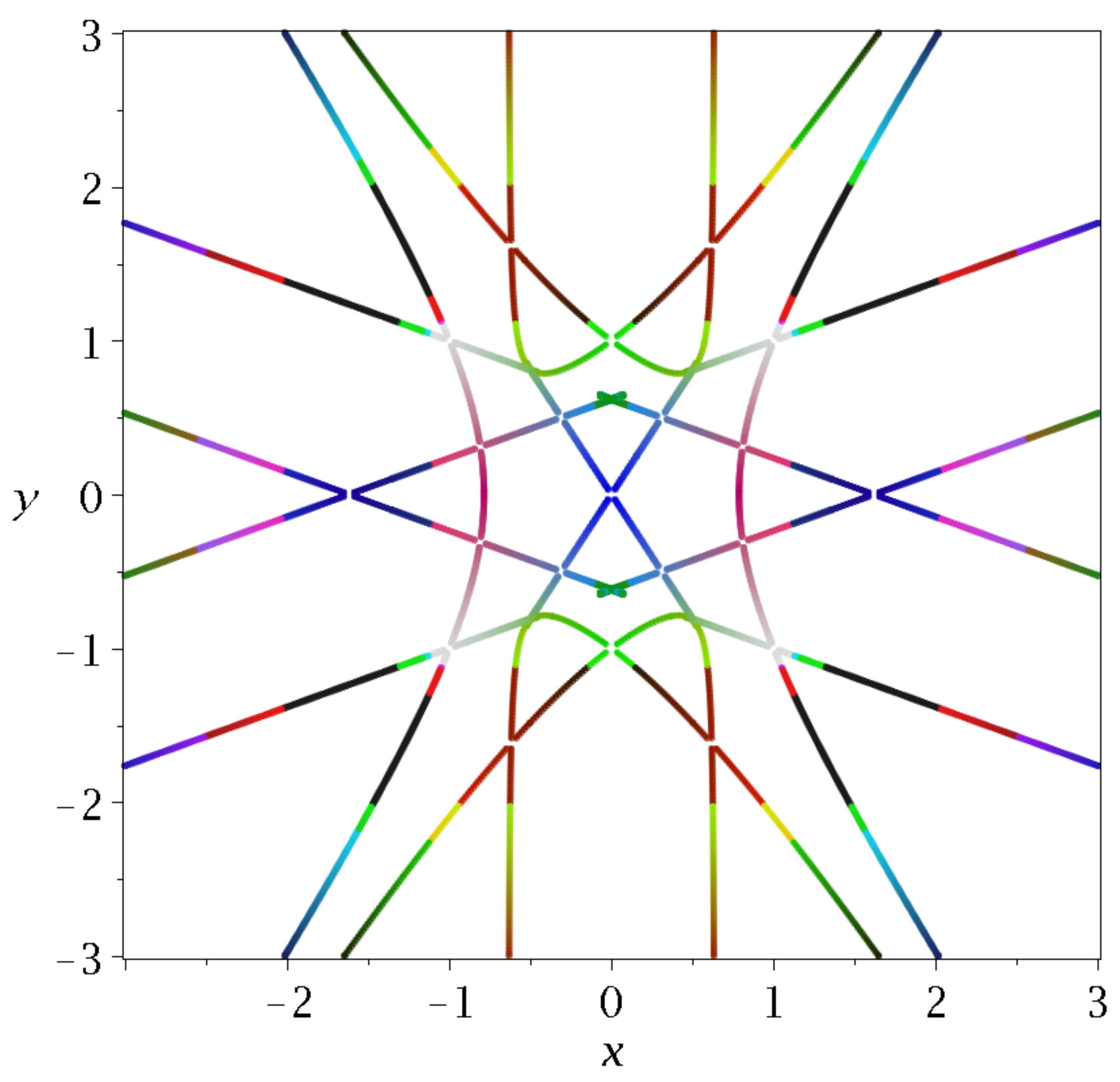}
{Projection on $(x,y)$ space.}
\end{minipage}
   \caption{Visualization of Example~\ref{ex:BarthSextic}.}
  \label{fig:BarthSextic}
\end{figure}

\begin{figure}
\begin{minipage}[t]{0.48\linewidth}
\centering
\includegraphics[width=\textwidth]{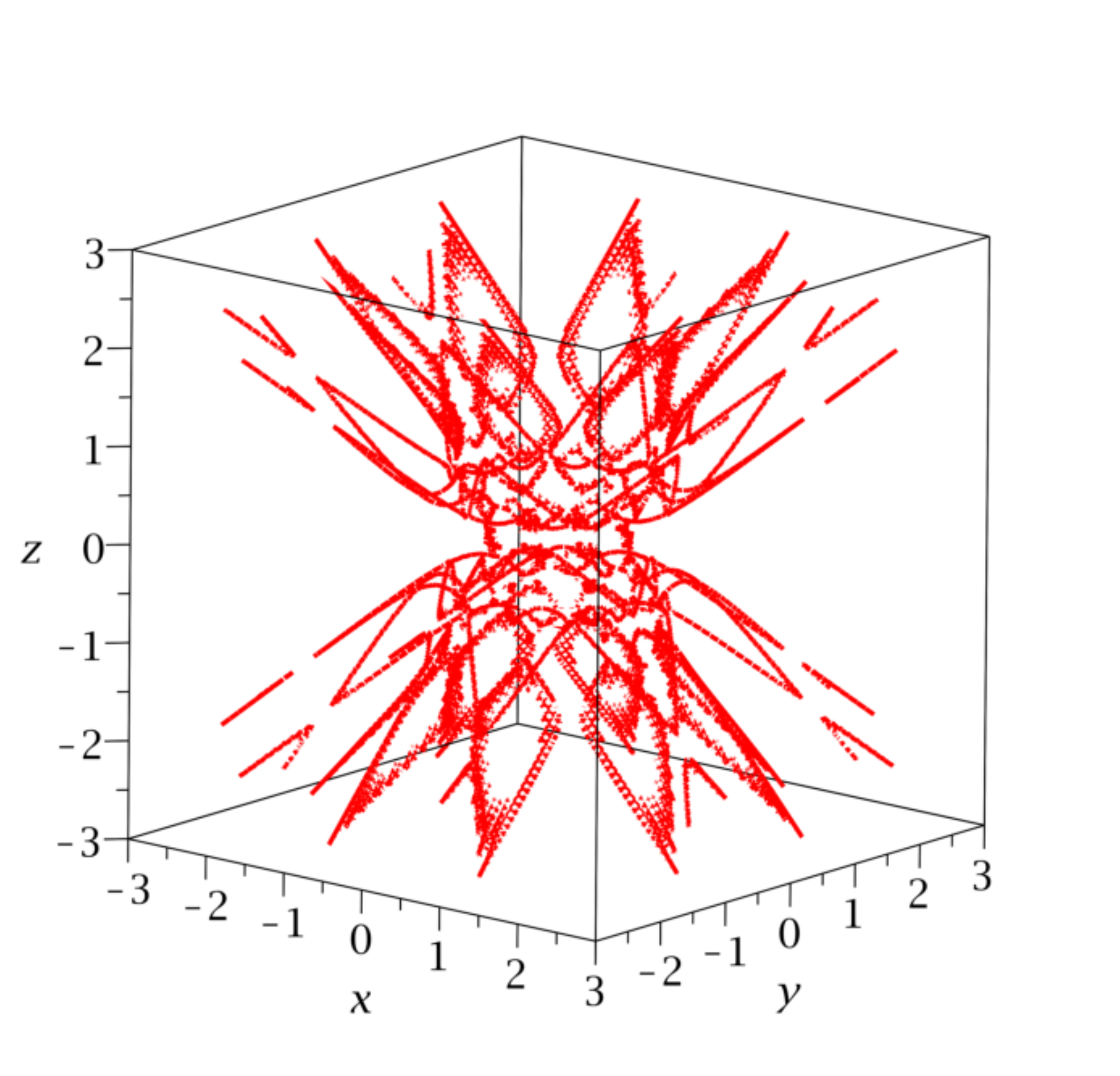}
{By Maple {\sf plots:-intersectplot}.}
\end{minipage}
\begin{minipage}[t]{0.5\linewidth}
\centering
\includegraphics[width=\textwidth]{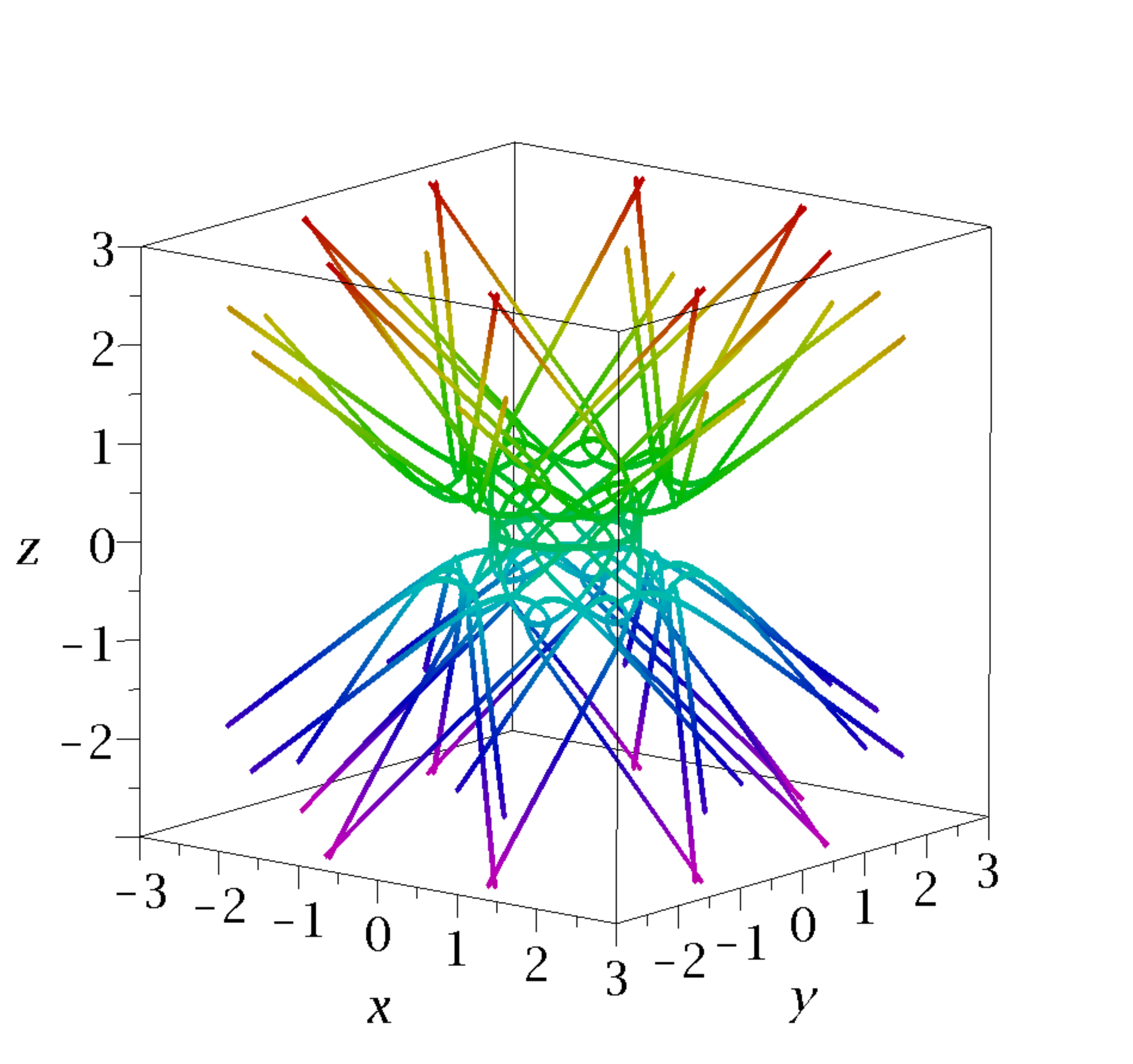}
{By {\sf ApproxPlot}.}
\end{minipage}
   \caption{Visualization of Example~\ref{ex:EndrassOctic}.}
  \label{fig:EndrassOctic}
\end{figure}

\section{Conclusion and future work}
\label{sec:con}
In this paper, we presented algorithms for visualizing
planar and space implicit algebraic curves with singularities.
The theoretical algorithm guarantees the polygonal approximation
$\epsilon$-close to the curve.
We introduced several strategies to turn the theoretical
algorithm to be practical and illustrate its effectiveness by examples.
One bottleneck of the algorithm is the computation of singular points,
whose efficiency might be improved if the curve is known to be
the resultant or discriminant of two polynomials~\cite{Imbach2017}.

The algorithm presented in this paper can also be used
for tracing space curves with singularities in ambient space
with dimension $>3$,
which is important for plotting border curves of parametric systems.
An efficient algorithm for computing isolated
singular points will be important for this method.

From a numeric point of view, singular points
are not stable w.r.t. perturbation.
A small perturbation may transform a singular point
to be exactly nonsingular but still
be ill-conditioned in the numerical sense.
We proposed heuristic strategies treating
these ``pseudo-singular" cases  and ``true-singular''
cases in the same away.
A limitation of current method is that we assume that the corank of the Jacobian matrix
is one.
A possible direction to remove such assumptions is to generalize the penalty method
for computing witness points in~\cite{Wu2017}
to tracing curves.

\acknowledgements{\rm The authors would like to thank Chee K. Yap for many valuable suggestions,
Michael Monagan for his advice on usage of {\sf Plots:-implicitplot}, and anonymous referees for helpful suggestions and comments.}


\begin{thebibliography}{10}
\providecommand{\url}[1]{\texttt{#1}}
\providecommand{\urlprefix}{URL }


\bibitem{Bresenham1977}
Bresenham J,
\newblock A linear algorithm for incremental digital display of circular arcs,
\newblock 1977, 
\newblock {\em Commun. ACM}, {\bf 20}(2):100--106.

\bibitem{Chandler1988}
Chandler R,
\newblock A tracking algorithm for implicitly defined curves,
\newblock 1988, 
\newblock {\em IEEE Comput. Graph. Appl.}, {\bf 8}(2):83--89.

\bibitem{Hong1996}
Hong H,
\newblock An efficient method for analyzing the topology of plane real
algebraic curves,
\newblock 1996, 
\newblock {\em Mathematics and Computers in Simulation}, {\bf 42}(4):571 -- 582.

\bibitem{Lopes2002}
  Lopes H, Oliveira J, and Figueiredo L,
\newblock Robust adaptive polygonal approximation of implicit curves,  
\newblock 2002, 
\newblock {\em Computers \& Graphics}, {\bf 26}(6):841 -- 852.

\bibitem{Sagraloff2009}
Emeliyanenko P, Berberich E, and Sagraloff M,
\newblock Visualizing arcs of implicit algebraic curves, exactly and fast,
\newblock {\em Advances in Visual Computing}, 2009, 608--619.

\bibitem{Labs2010}
  O.~Labs,
\newblock A List of Challenges for Real Algebraic Plane Curve
  Visualization Software,
\newblock {\em Nonlinear Computational Geometry}, 2010, 137--164.

\bibitem{Burr2012}
Burr M, Choi S, Galehouse B, and Yap C,
\newblock Complete subdivision algorithms, ii: Isotopic meshing of singular
algebraic curves,
\newblock 2012,
\newblock {\em Journal of Symbolic Computation}, {\bf 47}(2):131 -- 152.

\bibitem{Jin2015}
Jin K, Cheng J, and Gao X,
\newblock On the topology and visualization of plane algebraic curves,
\newblock {\em Proceedings of CASC}, 2015, 245--259.

\bibitem{Daouda2008}
Daouda D, Mourrain B, and Ruatta O,
\newblock On the computation of the topology of a non-reduced implicit space
  curve,
\newblock {\em Proceedings of ISSAC}, 2018, 47--54.

\bibitem{Jin2014}
Jin K and Cheng J,
\newblock Isotopic epsilon-meshing of real algebraic space curves,
\newblock {\em Proceedings of SNC}, 2014, 118--127.

\bibitem{DBLP:phd/dnb/Stussak13}
Stussak C,
\newblock {\em On reliable visualization algorithms for real algebraic curves
  and surfaces},
\newblock PhD thesis, Universit{\"{a}}ts- und Landesbibliothek Sachsen-Anhalt,
  Halle (Saale), 2013.

\bibitem{Cheng2010}
Cheng J, Lazard S, Pe{\~{n}}aranda L, Pouget M, Rouillier F, and Tsigaridas E,
\newblock On the topology of real algebraic plane curves,
\newblock 2010,
\newblock {\em Mathematics in Computer Science}, {\bf 4}(1):113--137.

\bibitem{Seidel2005}
Seidel R and Wolpert N,
\newblock On the exact computation of the topology of real algebraic curves,
\newblock {\em Proceedings of the Twenty-first Annual Symposium on
  Computational Geometry}, 2005, 107--115.

\bibitem{Imbach2017}
Imbach R, Moroz G, and Pouget M,
\newblock A certified numerical algorithm for the topology of resultant and
discriminant curves,
\newblock 2017,
\newblock {\em Journal of Symbolic Computation}, {\bf 80}:285 -- 306.

\bibitem{Gomes2014}
Gomes A,
\newblock A continuation algorithm for planar implicit curves with
singularities,
\newblock 2014,
\newblock {\em Computers \& Graphics}, {\bf 38}:365 -- 373.

\bibitem{DBLP:journals/toms/BrakeBHHSW17}
Brake D, Bates D, Hao W, Hauenstein J, Sommese A, and Wampler C,
\newblock Algorithm 976: Bertini{\_}real: Numerical decomposition of real
algebraic curves and surfaces,
\newblock 2017,
\newblock {\em {ACM} Trans. Math. Softw.}, {\bf 44}(1):10:1--10:30.

\bibitem{col75}
Collins G,
\newblock Quantifier elimination for real closed fields by cylindrical
  algebraic decompostion,
\newblock {\em Automata Theory and Formal Languages 2nd
  GI Conference}, 1975, 134--183.

\bibitem{Rouillier2000}
Rouillier F, Roy M, and {Safey El Din} M,
\newblock Finding at least one point in each connected component of a real
algebraic set defined by a single equation,
\newblock 2000,
\newblock {\em Journal of Complexity}, {\bf 16}(4):716 -- 750.

\bibitem{Chen2013}
Chen C, Davenport J, May J, {Moreno Maza} M, Xia B, Xiao R,
\newblock Triangular decomposition of semi-algebraic systems,
\newblock {\em J. Symb. Comp.}, {\bf 49}:3--26, 2013.

\bibitem{Hauenstein2012}
Hauenstein J,
\newblock Numerically computing real points on algebraic sets,
\newblock 2012,
\newblock {\em Acta Applicandae Mathematicae}, {\bf 125}(1):105--119.

\bibitem{WRF2017}
Wu W, Reid G, and Feng Y,
\newblock Computing real witness points of positive dimensional polynomial
systems,
\newblock 2017,
\newblock {\em Theoretical Computer Science}, {\bf 681}:217 -- 231.

\bibitem{Wu2017}
Wu W, Chen C, and Reid G,
\newblock Penalty function based critical point approach to compute real
  witness solution points of polynomial systems,
\newblock {\em Proceeding of CASC}, 2017, 377--391.

\bibitem{Blum1997}
Blum L, Cucker F, Shub M, and Smale S,
\newblock {\em Complexity and Real Computation},
\newblock Springer-Verlag, New York, 1998.

\bibitem{Beltran2013}
Beltr{\'a}n C and Leykin A,
\newblock Robust certified numerical homotopy tracking,
\newblock 2013,
\newblock {\em Foundations of Computational Mathematics}, {\bf 13}(2):253--295.

\bibitem{Martin2013}
Martin B, Goldsztejn A, Granvilliers L, and Jermann C,
\newblock Certified parallelotope continuation for one-manifolds,
\newblock 2013,
\newblock {\em SIAM Journal on Numerical Analysis}, {\bf 51}(6):3373--3401.

\bibitem{Yu2014}
Yu Y, Yu B, and Dong B,
\newblock Robust continuation methods for tracing solution curves of
parameterized systems,
\newblock 2014,
\newblock {\em Numerical Algorithms}, {\bf 65}(4):825--841.

\bibitem{Bajaj1997}
Bajaj C and Xu G,
\newblock Piecewise rational approximations of real algebraic curves,
\newblock 1997,
\newblock {\em Journal of Computational Mathematics}, {\bf 15}(1):55--71.

\bibitem{Leykin2006}
Leykin A, Verschelde J, and Zhao A,
\newblock Newton's method with deflation for isolated singularities of
polynomial systems,
\newblock 2006,
\newblock {\em Theoretical Computer Science}, {\bf 359}(1):111 -- 122.

\bibitem{Chen2017}
Chen C, Wu W, and Feng Y,
\newblock Full rank representation of real algebraic sets and applications,
\newblock {\em Proceeding of CASC}, 2017, 51--65.

\bibitem{Chen2016}
Chen C and Wu W,
\newblock A numerical method for computing border curves of bi-parametric real
  polynomial systems and applications,
\newblock {\em Proceeding of CASC}, 2016, 156--171.

\bibitem{Chen2016b}
Chen C and Wu W,
\newblock A numerical method for analyzing the stability of bi-parametric
  biological systems,
\newblock {\em Proceeding of SYNASC}, 2016, 91--98.

\bibitem{YX05}
Yang L and Xia B,
\newblock Real solution classifications of a class of parametric semi-algebraic
  systems,
\newblock {\em Proceeding of A3L}, 2005, 281--289.

\bibitem{LazardRouillier2007}
Lazard D and Rouillier F,
\newblock Solving parametric polynomial systems,
\newblock 2007,
\newblock {\em J. Symb. Comput.}, {\bf 42}(6):636--667.

\bibitem{Bennett2016}
Bennett H, Papadopoulou E, and Yap C,
\newblock Planar minimization diagrams via subdivision with applications to
anisotropic voronoi diagrams,
\newblock 2016, 
\newblock {\em Eurographics Symposium on Geometry Processing}, {\bf 35}(5):229--247.

\bibitem{ChenW18}
Chen C and Wu W,
\newblock A continuation method for visualizing planar real algebraic curves
with singularities,
\newblock 2018,
\newblock {\em Proceeding of {CASC}}, 2018, 99--115.

\bibitem{WR13}
Wu W and Reid G,
\newblock Finding points on real solution components and applications to
  differential polynomial systems.
\newblock {\em Proceeding of ISSAC}, 2013, 339--346.

\bibitem{Stewart90}
Stewart G,
\newblock Perturbation theory for the singular value decomposition,
\newblock {\em SVD and Signal Processing, II: Algorithms, Analysis and
  Applications}, 1990, 99--109.

\bibitem{Shen2012}
Shen F, Wu W, and Xia B,
\newblock Real root isolation of polynomial equations based on hybrid
  computation,
\newblock {\em Proceeding of ASCM 2012}, 2012, 375--396.

\bibitem{Lee2008}
Lee T, Li T, and Tsai C,
\newblock Hom4ps-2.0: a software package for solving polynomial systems by the
polyhedral homotopy continuation method,
\newblock 2008,
\newblock {\em Computing}, {\bf 83}(2):109.

\bibitem{Duzhin}
Duzhin S and Weisstein E,
\newblock Ordinary double point,
\newblock {\em From MathWorld--A Wolfram Web Resource}.

\end{thebibliography}


\end{document}